\documentclass[lettersize,journal]{./styl_files/IEEEtran}

\usepackage{balance}

\usepackage{amstext}
\usepackage{amsmath}
\usepackage{amsfonts}
\usepackage{amsthm}
\usepackage{multicol}
\usepackage{placeins}
\usepackage{orcidlink}
\newtheorem{theorem}{Theorem}
\newtheorem{lemma}{Lemma}

\newenvironment{proofsketch}{\proof}{\endproof}

\makeatletter
\AtBeginDocument{%
	\def\doi#1{\url{https://doi.org/#1}}}
\makeatother

\usepackage{stmaryrd}
\usepackage{etoolbox}
\usepackage{amsmath,amsthm}
\usepackage{amsbsy}
\usepackage{amsfonts}

\usepackage{tcolorbox}
\definecolor{darkblue}{rgb}{0,0,0.6}
\definecolor{darkgreen}{rgb}{0,0.5,0}
\definecolor{maroon}{rgb}{0.5,0.1,0.1}
\definecolor{dpurple}{rgb}{0.2,0,0.65}
\definecolor{darkviolet}{RGB}{130,95,141}
\definecolor{darkkhaki}{RGB}{189,183,107}
\definecolor{airforceblue}{rgb}{0.36, 0.54, 0.66}
\definecolor{ceruleanblue}{rgb}{0.16, 0.32, 0.75}
\definecolor{darkcandyapplered}{rgb}{0.64, 0.0, 0.0}
\definecolor{darkspringgreen}{rgb}{0.09, 0.45, 0.27}

	\usepackage{chngcntr}
 \usepackage{cleveref}
\usepackage[T1]{fontenc}
\usepackage{cases}

\usepackage{mdframed}
\usepackage{booktabs}
\usepackage{algpseudocode}
\usepackage[lined,linesnumbered,commentsnumbered]{algorithm}
\usepackage[font=normalsize,format=plain,labelfont={up},labelsep=space,justification=centering]{caption}

\usepackage{datenumber}
\usepackage{calc}
\newcounter{datetoday}
\newcounter{diffyears}
\newcounter{diffmonths}
\newcounter{diffdays}
\newcommand{\difftoday}[3]{%
      \setmydatenumber{datetoday}{\the\year}{\the\month}{\the\day}%
      \setmydatenumber{diffdays}{#1}{#2}{#3}%
      \addtocounter{diffdays}{-\thedatetoday}%
      \ifnum\value{diffdays}>0
        \def\diffbefore{Only }%
        \def\diffafter{left}%
      \else
        \def\diffbefore{}%
        \def\diffafter{ago}%
        \setcounter{diffdays}{-\value{diffdays}}%
      \fi
      \setcounter{diffyears}{\value{diffdays}/365}%
      \setcounter{diffdays}{\value{diffdays}-365*\value{diffyears}}%
      \setcounter{diffmonths}{\value{diffdays}/30}%
      \setcounter{diffdays}{\value{diffdays}-30*\value{diffmonths}}%
      \diffbefore
      \ifnum\value{diffyears}=0
      \else
        \ifnum\value{diffyears}>1
            {\color{red}\thediffyears\space years},
        \else
            {\color{red}\thediffyears\space year},
        \fi
      \fi
      \ifnum\value{diffmonths}=0
      \else
        \ifnum\value{diffmonths}>1
            {\color{red}\thediffmonths\space months}
        \else
            {\color{red}\thediffmonths\space month}
        \fi
      \fi
      {\color{red}\thediffdays\space days}
      \diffafter
}

\usepackage[scr=boondoxo]{mathalfa}

\usepackage{tikz}
\usetikzlibrary{positioning}
\usetikzlibrary{arrows.meta}
\usetikzlibrary{decorations.pathmorphing}
\tikzset{snake it/.style={decorate, decoration=snake}}

\DeclareMathAlphabet{\mathpzc}{T1}{pzc}{m}{it}

\newcommand{\matx}[1]{{\bf #1}} 
\newcommand{\vecc}[1]{{\bf #1}}

\newcommand{\win}[1]{X_{{#1}}}
\newcommand{\game}[1]{\mathsf{Game}_{{#1}}}

\newcommand{\prob}[1]{\mathrm{Pr}\left[{#1}\right]}

\def\P{\mathbf{P}}
\def\B{\mathbf{B}}

\def\SBGgen{\mathcal{G}_{\sf sbg}}

\newcommand*{\IndCPA}{\mathsf{IND{\mathchar`-}CPA}}

\newcommand*{\IndCCA}{\mathsf{IND{\mathchar`-}CCA}}
\newcommand*{\IndsCCA}{\mathsf{sIND{\mathchar`-}CCA}}

\def\qd{Q}

\def\ppt{\mathsf{ppt}}

\newcommand{\D}{\matx{D}}
\newcommand{\F}{\matx{F}}
\newcommand{\A}{\matx{A}}

\def\z{\mathbf{z}}

\def\r{\mathbf{r}}

\newcommand{\M}{{\matx{M}}}

\def\sample{\hookleftarrow}

\def\c{\vecc{c}}

\def\z{\vecc{z}}

\newcommand{\iCol}[2]{\left(\begin{matrix}#1\\#2\end{matrix}\right)}

\newcommand{\Sum}[2]{\sum\limits_{\substack{#1\in #2}}}

\newcommand{\Summ}[4]{\sum\limits_{\substack{#1\in #2\\ #3\in #4}}}
\newcommand{\Prod}[2]{\prod\limits_{\substack{#1\in #2}}}

\newcommand{\Prodneq}[4]{\prod\limits_{\substack{#1\in #2\\ #3\neq #4}}}

\newcommand{\abslt}[1]{\ensuremath{\left\lvert #1 \right\rvert}}
\newcommand{\Adv}[2]{\mathsf{Adv}_{#1}^{#2}(\secpp)}

\def\neglgbl{\mathsf{neg}(\secpp)}

\def\AA{\mathcal{A}}
\def\AC{\mathcal{C}}
\def\AB{\mathcal{B}}

\newcommand{\Zp}[1][]{\mathbb{Z}_{p_{#1}}}
\def\ZN{\mathbb{Z}_N}
\def\ZZ{\mathbb{Z}}
\def\NN{\mathbb{N}}

\newcommand{\Go}{\mathbb{G}_1}						
\newcommand{\Gt}{\mathbb{G}_2}						
\newcommand{\GT}{\mathbb{G}_{\text{T}}}				

\def\sbG{\mathbb{G}}
\def\sbGT{\mathbb{G}_{\text{T}}}	
\def\sbg{g}

\def\abG{\mathsf{G}}

\def\prm{p}

\newcommand{\cabgo}[2][\sbg]{{#1}_{#2}}

\def\sbgo{\cabgo[\sbg]{1}}
\def\sbgt{\cabgo[\sbg]{2}}
\def\sbgr{\cabgo[\sbg]{3}}
\def\sbgot{\cabgo[\sbg]{12}}

\def\sbgtr{\cabgo[\sbg]{23}}

\newcommand{\sbGi}[2][\sbG]{{#1}_{p_{#2}}}
\def\sbGo{\sbGi{1}}
\def\sbGt{\sbGi{2}}
\def\sbGr{\sbGi{3}}
\def\sbGot{\sbGi{1}\sbGi{2}}
\def\sbGor{\sbGi{1}\sbGi{3}}
\def\sbGtr{\sbGi{2}\sbGi{3}}

\def\dgho{\mathsf{SD1}}
\def\dght{\mathsf{SD2}}

\def\dejaq{\text{D\'ej\`a Q}}

\def\bee{\mathfrak{b}}

\def\kay{\mathscr{k}}								

\def\eu{\mathscr{u}}

\def\zee{\mathscr{z}}	
\def\ii{\mu}
\def\jj{\delta}

\def\IniT{\mathbf{Init}}
\def\SetuP{\mathbf{Setup}}
\def\QPhase#1{\mathbf{Query\ Phase}\text{-}{\sf #1}}
\def\KeyQ{\mathbf{Key\ Queries}}

\def\DecQ{\mathbf{Dec\ Queries}}	
	
\def\Chal{\mathbf{Challenge}}

\def\Setup{\mathsf{Setup}}
\def\Kgen{\mathsf{KeyGen}}

\def\Enc{\mathsf{Encrypt}}
\def\Dec{\mathsf{Decrypt}}

\def\msk{\mathsf{msk}}
\def\mpk{\mathsf{mpk}}

\def\sk{\mathsf{sk}}
\def\ssk{{\kappa}}

\def\ct{\mathsf{ct}}

\def\keyrand{r}

\newcommand{\kyr}[1][]{\keyrand_{#1}}

\newcommand{\hkyr}[1][]{\hat{\keyrand}_{#1}}

\def\cmsk{{\gamma}}
\def\bmsk{{\beta}}
\def\tbmsk{\widetilde{\bmsk}}
\def\amsk{{\alpha}}

\def\Sk{\mathcal{SK}}
\def\CT{\mathcal{CT}}
\def\Ct{\mathsf{Ct}}
\def\Ssk{\mathscr{K}}

\def\hfnc{\mathsf{H}}
\def\id{\mathsf{id}}
\def\ID{\mathcal{ID}}
\def\XSet{\mathsf{X}}
\def\YSet{\mathsf{Y}}
\def\chYSet{\mathsf{Y}^*}

\def\Xsz{\kay}
\def\Ysz{\ell}
\def\SSet{\mathsf{S}}

\def\USet{\mathcal{U}}
\def\Usz{\eu}

\mathchardef\mhyf="2D

\newcommand{\detr}[1]{\mathsf{det}(#1)}
\def\Span{\mathsf{span}}

\def\spe{\mathsf{SsPE}}

\def\ibr{\mathsf{IBR}}

\def\msg{M}

\def\MSg{\mathcal{M}}

\newcommand{\xpz}{{[0]}}
\newcommand{\xpo}{{[1]}}
\newcommand{\xpb}{{[\bee]}}

\def\etal{\textit{et al.\ }}

\def\abv#1#2{\overset{#1}{#2}}

\def\iseq{\abv{?}{=}}

\def\Au{\vecc{a}}
\def\Fu{\vecc{f}}
\def\Ad{\underline{\A}}

\def\rowcnto{m}
\def\rowcntt{q}
\def\ttee{t}
\def\colcnt{\rowcnto+\rowcntt+1}
\def\cnt{m+\ttee+1}

\def\speo{\mathsf{CM\mhyf I}}
\def\spet{\mathsf{CM\mhyf II}}
\def\mpk{\mathsf{pp}}
\def\secpp{\lambda}

\begin{document}

\title{\hspace{-0.45cm}Large Universe Subset Predicate Encryption with IND-CCA Security \\ (with Constant-size Ciphertext and Keys)}

\author{
    \IEEEauthorblockN{Sayantan Mukherjee\orcidlink{0000-0002-1236-8086}}\\
    \IEEEauthorblockA{\textit{Department of Computer Science and Engineering} \\
        \textit{Indian Institute of Technology Jammu}\\
            Jammu, India \\
            csayantan.mukherjee@gmail.com
    }
}

\markboth{Journal of \LaTeX\ Class Files,~Vol.~14, No.~8, August~2021}%
{Shell \MakeLowercase{\textit{et al.}}: Large Universe Subset Predicate Encryption with IND-CCA Security \\ (with Constant-size Ciphertext and Keys)}



\maketitle



\begin{abstract}
 Katz et al. (CANS'17) introduced Subset Predicate Encryption (SPE). 
 This scheme is a generalization of broadcast encryption as it emulates the \emph{subset containment} predicate in the encrypted domain. 
 They proposed two selectively IND-CPA secure SPE constructions in the small universe setting. 
 They also showed some black-box transformations of SPE to well-known primitives like WIBE and ABE to establish the richness of the SPE structure.\\
 Chatterjee and Mukherjee (RSA'19) proposed two SPE constructions in the large-universe setting.
 Their first construction achieved constant-size ciphertexts and secret keys, but it is proven secure in a restricted version of selective security. 
 Although the second construction achieves adaptive security, the ciphertext size depends on the size of the data-attribute set.
 Furthermore, neither of these two constructions achieves CCA security.\\
 In this work, we propose the first large-universe CCA-secure subset predicate encryption with constant-size ciphertext and secret keys.
 We prove this construction achieves standard selective security under the standard subgroup decision problems.
 Finally, we transform our extremely efficient SPE into the first CCA-secure WIBE, WKD-IBE, etc., with constant-size ciphertexts and secret keys via black-box transformations.
\end{abstract}
    
    

\section{Introduction} \label{sec:Intro}

Broadcast Encryption (BE) \cite{C:BonGenWat05} is a well-explored primitive that allows the encryption of a message for a set of users.
In a BE scheme, a ciphertext for a data-attribute set $\YSet$ can be decrypted by a secret key on an identity $x$ if $x\in\YSet$.
Katz \etal \cite{CANS:KMMS17} proposed an extension of this primitive, namely Subset Predicate Encryption (SPE).
Here, an SPE ciphertext for a data-attribute set $\YSet$ can be decrypted by a secret key for another key-attribute set $\XSet$ if $\XSet\subseteq\YSet$.
Katz \etal \cite{CANS:KMMS17} then introduced several generic transformations to convert an SPE to Wildcarded IBE (WIBE) \cite{ICALP:ACDMNS06}, Wicked IBE (WKD-IBE) \cite{ESORICS:AbdKilNev07}, Ciphertext Policy Attribute-based Encryption (CP-ABE) for Disjunctive Normal Form (DNF), etc.
They also proposed two selective CPA-secure SPE constructions in the small-universe setting.

A follow-up work by Chatterjee and Mukherjee \cite{RSA:ChaMuk19} proposed two large-universe\footnote{The universe is exponential-size.} SPE constructions.
Their first construction $(\speo)$ achieved constant-size ciphertexts and secret keys, but is proven secure in a restrictive variant of selective IND-CPA security, which they referred to as selective$^*$ IND-CPA security.
To achieve adaptively IND-CPA security, their second construction $(\spet)$ requires a ciphertext length that is dependent on the size of the set $\YSet$.
Later works \cite{DSC:TseGao21,CSI:Tseng24} also aimed at SPE construction where the ciphertext is dependent on $|\YSet|$.

A separate line of work \cite{PKC:YAHK11,PKC:YASSHK12,AAECC:NanPan17} studied the problem of transforming an $\IndCPA$ secure predicate encryption to an $\IndCCA$ secure predicate encryption \cite{EC:KatSahWat08}.
They utilize different frameworks, such as verifiability and delegation, which are often less efficient or more challenging to achieve.
Follow-up works \cite{INDOCRYPT:ChaMukPan17,AMC:NanPan22} proposed a simpler \emph{direct} technique to get a CCA-secure predicate encryption construction.
However, to date, the efficient \emph{direct} approach has only been applied to pair-encoding-based predicate encryption \cite{EC:Attrapadung14}.
We emphasize that the \emph{direct} approach of \cite{INDOCRYPT:ChaMukPan17,AMC:NanPan22} utilizes the linear reconstruction-based structure critically.

\subsection{Motivation}
IBE with wildcards have been studied for some time. As per our knowledge, all the existing CPA secure solutions grow with $(i)$ the number of non-wildcard components \cite{ICALP:ACDMNS06,ESORICS:AbdKilNev07} or $(ii)$ with the number of wildcard components \cite{SCN:SVNHJ10}. Katz \etal \cite{CANS:KMMS17} and follow-up work by Chatterjee and Mukherjee \cite{RSA:ChaMuk19} also achieve similar efficiency. Achieving reasonably secure and efficient constructions of IBE with wildcards has been a long-standing open question. Of course, one could generically transform $\speo$, the first SPE construction of Chatterjee and Mukherjee \cite{RSA:ChaMuk19}, towards an IBE with wildcard with constant-size ciphertexts and keys. However, $\speo$ achieved security under the selective$^*$ IND-CPA security model, which is highly restrictive for practical purposes. Moreover, they did not consider the question of CCA security, which is essential for real-world deployment. 

\subsubsection{Commentary on Selective$^*$ IND-CPA security of \cite{RSA:ChaMuk19}}
The highly restricted nature of selective$^*$ IND-CPA security of \cite{RSA:ChaMuk19} sips into the constructions of IBE with wildcards as well.
We justify this claim by first explaining the selective$^*$ IND-CPA security of \cite{RSA:ChaMuk19} followed by its effect on the transformations that we described in \Cref{sec:Applications}. 
The selective$^*$ IND-CPA security of \cite{RSA:ChaMuk19} allows the adversary to challenge on $\chYSet$ even before getting the public param $\mpk$. 
The adversary can then make polynomially many subset queries for $\XSet_i$ adaptively, such that $\XSet_i\not\subset \YSet$ and $\XSet_i\setminus (\cup_{j\in[q]\setminus\{i\}} \XSet_j)\neq \phi$.
To elaborate, the second condition requires that $\forall i\in[q]$, $\exists x\in\XSet_i$ s.t. $x\notin \cup_{j\in[i-1]} \XSet_j$.
Informally speaking, every $i^{th}$ key generation query should have a completely new element that is absent in all key index sets for all previous key generation queries.
The transformations that we described in \Cref{sec:Applications} encode identities into sets. 
For example, key queries on $(1,0,0)$, $(1,1,0)$, $(1,0,1)$ and $(1,1,1)$ will encode into $\{1,4,6\}$, $\{1,3,6\}$, $\{1,4,5\}$ and $\{1,3,5\}$ respectively.
Observe that, all elements of $\{1,3,5\}$ are present in $\{1,4,6\}\cup\{1,3,6\}\cup\{1,4,5\}$.
Thus, due to the limitation of underlying selective$^*$ IND-CPA security, $\speo$-based constructions of IBE with wildcards cannot allow key queries on identities $(1,0,0)$, $(1,1,0)$, $(1,0,1)$ and $(1,1,1)$. 
Thus the adversary of selective$^*$ IND-CPA security of \cite{RSA:ChaMuk19} is unnaturally restricted.

\subsection{Our Contribution}
In this work, we propose the first large-universe SPE construction with constant-size ciphertext and secret key that is proven selectively IND-CCA secure. We emphasize that our SPE construction achieves optimal ciphertexts and secret keys, even in the exponential-sized universe. To summarise, our achievement here is threefold:
\begin{itemize}
    \item We critically analyse $\speo$ to remove the restriction on the selective IND-CPA security that \cite{RSA:ChaMuk19} imposed.
    \item We modify $\speo$ to achieve IND-CCA security, paying an insignificant price. Looking ahead, our construction adds only one additional group element to the ciphertext and four group elements to the public parameter over $\speo$.
    \item By applying the generic transformations from \cite{CANS:KMMS17} to our SPE construction, we obtain CCA-secure WIBE, wicked IBE (WKD-IBE), and CP-ABE for disjunction (DNF-ABE) with constant-size ciphertexts and secret keys. This resolves a long-standing question of IBE with wildcards using a constant-size ciphertext and secret key. 
\end{itemize}

\subsection{Technical Overview}
We begin with the security proof of $\speo$ \cite{RSA:ChaMuk19}, following a brief overview of its construction, simplified to its core.
To achieve constant-size ciphertext and keys, $\speo$ \cite{RSA:ChaMuk19} encodes both data-attribute set $\YSet$ and key-attribute set $\XSet$ in terms of corresponding characteristic polynomial $P_{\YSet}(\zee)$ and $P_{\XSet}(\zee)$ respectively where $P_{\SSet}(\zee)=\Prod{x}{\SSet} (\zee+x)$.
A secret key $u^{\frac{1}{P_{\XSet}(\zee)}}$ of $\speo$ decrypts a ciphertext $g^{P_{\YSet}(\zee)}$ by computing $e(g,u)^{P_{\YSet\setminus\XSet}(\zee)}$ where $g,u\in\sbG$.
Note that we provide this simplified representation only to informally overview the functionality.
Secret keys and ciphertexts of the actual $\speo$ \cite{RSA:ChaMuk19} define additional components involving randomness for security.

We now focus on the security proof of $\speo$ \cite{RSA:ChaMuk19} informally.
Chatterjee and Mukherjee \cite{RSA:ChaMuk19} argued the proof of security of $\speo$ via a standard $\dejaq$-based argument \cite{TCC:Wee16,SCN:GonLibRam18}, which is a dual-system encryption-based proof technique \cite{C:Waters09}.
Informally speaking, the proof makes the challenge ciphertext \emph{semi-functional} and then gradually makes the secret keys \emph{semi-functional}. 
Once the challenge ciphertext and all the keys are semi-functional, the proof replaces the semi-functional components with random elements and argues that such a change will be invisible to even an unbounded adversary.
To make this replacement invisible to the adversary, the proof expresses the semi-functional components as a system of linear equations $(\z=\A\r)$. 
It is sufficient to argue that the linear map $\A$ is non-singular for the proof to go through.
However, this is where Chatterjee and Mukherjee \cite{RSA:ChaMuk19} encounter a limitation, which we explain in the following section.

To explain this, we recall that each row of $\A$ is essentially the semi-functional component of secret keys queried or public parameters. 
Furthermore, we recall that each secret key of $\speo$ encoded a set $\XSet$ in terms of the characteristic polynomial $P_{\XSet}(\zee)=\Prod{x}{\XSet} (\zee+x)$.
Now, suppose the adversary makes a sequence of key extraction queries on $\XSet_1=\{a,b\},\XSet_2=\{b,c\},\XSet_3=\{a,c\}$ such that $\A[1],\A[2],\A[3]$ store $\frac{1}{P_{\XSet_1}(\zee)},\frac{1}{P_{\XSet_2}(\zee)},\frac{1}{P_{\XSet_3}(\zee)}$ respectively.
For $\XSet'=\{a,b,c\}$, it is easy to see that $\frac{1}{P_{\XSet'}(\zee)}$ can be computed via two paths here: $(i)$ pair $\frac{1}{P_{\XSet_1}(\zee)},\frac{1}{P_{\XSet_2}(\zee)}$, $(ii)$ pair $\frac{1}{P_{\XSet_2}(\zee)},\frac{1}{P_{\XSet_3}(\zee)}$. 
In particular, $\frac{1}{(\zee+a)(\zee+b)(\zee+c)} = \frac{1}{(c-a)}\cdot\frac{1}{(\zee+a)(\zee+b)}+\frac{1}{a-c}\cdot\frac{1}{(\zee+b)(\zee+c)}=\frac{1}{(b-a)}\cdot\frac{1}{(\zee+a)(\zee+c)}+\frac{1}{a-b}\cdot\frac{1}{(\zee+b)(\zee+c)}$.
For the given sequence of queries, it is easy to see that $\A$ is non-singular.
To stop these queries, they \cite{RSA:ChaMuk19} restricted the selective $\IndCPA$ security to the so-called selective$^*$ $\IndCPA$ security.

We first notice that the non-singularity of $\A$ is a sufficient condition for the proof but not a necessary one.
Looking ahead, we observe that a particular row (of our interest) of $\A$ is independent of the rest of the rows of $\A$.
Informally speaking, this particular row can be replaced with a random quantity irrespective of the query sequences any adversary makes. 
Interestingly, this small observation makes our security argument significantly more straightforward and easier to comprehend.
For better understanding, we give an informal proof sketch in \Cref{sec:SPE-Sec}.

To make this construction $\IndCCA$ secure, we revisit the so-called \emph{direct} technique of  \cite{INDOCRYPT:ChaMukPan17,AMC:NanPan22,INDOCRYPT:BlaMuk20}.
Looking ahead, we will also gradually modify the verification of all decryption queries.
The novelty in our proof is that we broaden the horizon of \emph{direct} technique over \cite{INDOCRYPT:ChaMukPan17,AMC:NanPan22,INDOCRYPT:BlaMuk20} which work in pair-encoding-based predicate encryptions only. 
We describe our approach next informally.
We introduce an extra ciphertext component, namely, a Boneh-Boyen Hash \cite{C:BonBoy04} evaluated on a tag $(\tau)$ where the tag $\tau$ is a collision-resistant hash of the existing ciphertext (which we sometimes call CPA-ciphertext).
This ensures that any change on this part of the ciphertext results in a different $\tau'$, and the universal hash argument makes them independent.
Based on this observation, we make the consistency check \emph{semi-functional}, which catches if the adversary makes a decryption query on the semi-functional ciphertext.
This ensures that a decryption query is answered correctly only if the adversary computed the ciphertext.

We now informally mention the difference of our technique with the so-called \emph{direct} technique of \cite{INDOCRYPT:ChaMukPan17,AMC:NanPan22}.
To retrieve the plaintext, \cite{INDOCRYPT:ChaMukPan17,AMC:NanPan22} used a pairing product equation which implicitly performed a consistency check.
Thus, if a ciphertext does not satisfy the consistency check, its decryption function does not produce the plaintext.
This is stronger than the requirement.
The Cramer-Shoup CCA-secure encryption inspires us \cite{C:CraSho98} where CPA-ciphertext decryption and consistency checks are somewhat loosely connected.
Like Cramer-Shoup \cite{C:CraSho98}, our decryption function is defined by two separate steps -- $(i)$ Perform consistency check, and abort if it fails; $(ii)$ Perform CPA decryption on the CPA-ciphertext.
Note that this technique allows a decryptor to retrieve the plaintext decrypting CPA-ciphertext if somehow the decryptor manages to avoid the consistency check.
Similar to Cramer-Shoup \cite{C:CraSho98}, we achieve CCA security by ensuring that decryption runs the consistency check first and aborts if it fails.
To summarise, we modify the so-called \emph{direct} technique of \cite{INDOCRYPT:ChaMukPan17,AMC:NanPan22} to support non-linear encodings as well.
This process opens up further application of this technique for other predicate encryptions that are not pair-encoding based, which was necessary for \cite{INDOCRYPT:ChaMukPan17,AMC:NanPan22}.

\paragraph{Organization of the Paper.}
In~\Cref{sec:Preliminaries}, we recall a few definitions, mathematical tools and present the notations that will be followed in this paper.
In~\Cref{sec:Construct}, we propose the first CCA-secure subset predicate encryption and its proofs.
\Cref{sec:Applications} presents the effect of our SPE construction and \Cref{sec:Conclude} concludes this paper.

\section{Mathematical Tools and Preliminaries} \label{sec:Preliminaries}

\subsection{Mathematical Tools}
\paragraph{Notations.} Here we denote $[a,b]=\{i\in\NN:a\leq i\leq b\}$ and for any $n\in\NN$, ${\small [n]=[1,n]}$.
By $s\sample S$ we denote a uniformly random choice $s$ from $S$. 
The security parameter is denoted by $1^{\secpp}$ where $\secpp\in\NN$.
We use $\neglgbl$ to denote negligible function.
We use 
$\Adv{\AA}{{\sf HP}}$ is used to denote the advantage of $\AA$ to solve the problem ${\sf HP}$.

Gong \etal \cite{SCN:GonLibRam18} proposed and proved the following lemma. We state it as a fact and use this result in this work.
\begin{lemma}\label{lem:SPE-I-GLR18}
    For $m,\ttee,\cmsk_1,\ldots,\cmsk_{m+\ttee+1}\in\NN$, if 
    {\scriptsize
    \begin{equation}
        \D = 
        \left[\begin{matrix}
            1 & 1 & \ldots & 1 \\
            \cmsk_1^{} & \cmsk_2^{} & \ldots & \cmsk_{m+\ttee+1}^{}\\\\
            \cmsk_1^2 & \cmsk_2^2 & \ldots & \cmsk_{m+\ttee+1}^2\\
            \vdots & \vdots & \ddots & \vdots \\
            \cmsk_1^m & \cmsk_2^m & \ldots & \cmsk_{m+\ttee+1}^m\\\\
            \frac{1}{(\cmsk_1+x_1)} & \frac{1}{(\cmsk_2+x_1)} & \ldots & \frac{1}{(\cmsk_{m+\ttee+1}+x_1)}\\\\
            \frac{1}{(\cmsk_1+x_2)} & \frac{1}{(\cmsk_2+x_2)} & \ldots & \frac{1}{(\cmsk_{m+\ttee+1}+x_2)}\\
            \vdots & \vdots & \ddots & \vdots \\
            \frac{1}{(\cmsk_1+x_\ttee)} & \frac{1}{(\cmsk_2+x_\ttee)} & \ldots & \frac{1}{(\cmsk_{m+\ttee+1}+x_\ttee)}\\
        \end{matrix}\right]
        \end{equation}
    }
    then $\detr{\D}=\delta\cdot\frac{\prod_{1\leq l<j\leq \ttee}(x_l-x_j)\prod_{1\leq i<k\leq \cnt} (\gamma_i-\gamma_k)}{\prod_{k=1}^{\cnt} \prod_{l=1}^{\ttee} (\gamma_k+x_l)}$ where $\delta\in\NN$. 
\end{lemma}


\subsection{Composite Order Bilinear Pairings.}\label{sec:Prelim-comp-biPairing}
A composite order symmetric bilinear group generator $\SBGgen$, takes the security parameter $1^{{\secpp}}$ and returns a $4$-tuple $(N,\sbG,\sbGT,e)$ where both $\sbG,\sbGT$ are cyclic groups of order $N=\prm_1\prm_2\prm_3$ where all $\prm_i$ are large primes and $e:\sbG\times\sbG\rightarrow \sbGT$ is an admissible, non-degenerate Type-1 bilinear pairing.
In this work, we denote $\sbG=\sbGo\sbGt\sbGr$ where $\sbGi{i}$ denotes a subgroup of $\sbG$ of order $p_i$.
By convention $\sbg_{i,j}$ is an element of subgroup $\sbGi{i}\sbGi{j}$ and $\sbGi{i}^*=\sbGi{i}\setminus \{1\}$.
It is evident that $e(\sbg_i, \sbg_j)=1$ if $i\neq j$.

\subsection{Hardness Assumptions}\label{sec:Prelim-comp-Hard}\label{sec:Prelim-Hard}
Let $(N,\sbG,\sbGT,e)\leftarrow \SBGgen(1^{\secpp})$ be the output of symmetric bilinear group generator where both $\sbG,\sbGT$ are cyclic groups of order $N=\prm_1\prm_2\prm_3$ where $\prm_1, \prm_2,\prm_3$ are large primes.
We define two variants of subgroup decision problems \cite{TCC:Wee16} as follows:

\paragraph{$\dgho$} 
The advantage of any adversary $\AA$ to solve the $\dgho$ is $\Adv{\AA}{\dgho}=$
{\[
    |\prob{\AA(\sbgo, \sbgr, \sbgot, T_0)\rightarrow 1}-\prob{\AA(\sbgo, \sbgr, \sbgot, T_1)\rightarrow 1}|
\] }
where $T_0\sample \sbGo$ and $T_1\sample \sbGot$ where $\sbgo\sample \sbGo^*$, $\sbgr\sample \sbGr^*$ and $\sbgot\sample \sbGot$.  
We say $\dgho$ is hard if $\Adv{\AA}{\dgho}\leq \neglgbl$ for any arbitrary $\ppt$ adversary $\AA$.

\paragraph{$\dght$} 
The advantage of any adversary $\AA$ to solve the $\dght$ is $\Adv{\AA}{\dght}=$
\[\hspace{-0.2cm}
    |\prob{\AA(\sbgo, \sbgr, \sbgot, \sbgtr, T_0)\rightarrow 1}-\prob{\AA(\sbgo, \sbgr, \sbgot, \sbgtr, T_1)\rightarrow 1}|
\] 
where $T_0\sample \sbGor$ and $T_1\sample \sbG$ where $\sbgo\sample \sbGo^*$, $\sbgr\sample \sbGr^*$, $\sbgot\sample \sbGot$ and $\sbgtr\sample\sbGtr$. 
We say $\dght$ is hard if $\Adv{\AA}{\dght}\leq \neglgbl$ for any arbitrary $\ppt$ adversary $\AA$.

\subsection{Subset Predicate Encryption (SPE)}\label{sec:SPE}
A subset predicate encryption (SPE) scheme is defined by a tuple of four $\ppt$ algorithms $(\Setup,\Kgen,\Enc,\Dec)$ as following.
\subsubsection{Definition}\label{sec:SPE-Def}
\begin{itemize}
    \item $\Setup(1^{\secpp},m)$: It takes the security parameter $1^{\secpp}$ and $m\in\mathsf{poly}({\secpp})$. It computes a master secret key $\msk$ and the corresponding public param $\mpk$. We assume $\mpk$ to store an identity space $\USet$, a message space $\MSg$, a key space $\Sk$ and a ciphertext space $\CT$.
    \item $\Kgen(\msk,\XSet)$: It takes $\msk$ and a set $\XSet\subset\USet$ s.t. $|\XSet|\leq m$. It outputs a secret key $\sk\in\Sk$.
    \item $\Enc(\mpk,\YSet,\msg)$: It takes $\mpk$, a set $\YSet\subset\USet$  s.t. $|\YSet|\leq m$ and a message $\msg\in\MSg$. It outputs a ciphertext $\Ct\in\CT$.
    \item $\Dec(\mpk,\sk,\Ct)$: It takes $\mpk$, $\sk$ and $\Ct$ as input. It outputs $\msg'\in \MSg\sqcup \{\bot\}$. The special symbol $\bot$ indicates rejection.
\end{itemize}

\paragraph{Correctness.} 
For any security parameter $1^{\secpp}$, any key-attribute set $\XSet\subset\USet$, any data-attribute set $\YSet\subset\USet$, such that if $\XSet\subseteq\YSet$,
\[\Pr[\Dec(\mpk,\sk,\Ct)=\msg]=1-\neglgbl\]
where the probability is taken over $(\msk,\mpk)\leftarrow\Setup(1^{\secpp},m)$, $\sk\leftarrow \Kgen(\msk,\XSet)$, $\Ct\leftarrow \Enc(\mpk,\YSet,\msg)$. 

\subsubsection{Security Definition.}\label{sec:SPE-SecDef}
The security of subset-predicate encryption $\spe$ is modeled as a security game between a challenger $\AC$ and an adversary $\AA$.
\begin{itemize}
    \item $\SetuP$: $\AC$ runs $\Setup$ to give out $\mpk$ and keeps $\msk$ as secret.
    
    \item $\QPhase{I}$: $\AA$ makes the following queries:
    \begin{itemize}
        \item $\KeyQ$: $\AA$ gives a key-attribute set $\XSet\subset\USet$ s.t. $|\XSet|\leq m$. $\AC$ returns $\sk\leftarrow\Kgen(\msk,\XSet)$ to $\AA$.
        \item $\DecQ$: $\AA$ gives a key-attribute set $\XSet\subset\USet$, a ciphertext $\Ct$ w.r.t. a data-attribute set $\YSet\subset\USet$ s.t. $|\XSet|,|\YSet|\leq m$. $\AC$ returns $\Dec(\mpk,\Kgen(\msk,\XSet),\Ct)$ to $\AA$.
    \end{itemize}
    
    \item $\Chal$: $\AA$ provides challenge messages $\msg^{\xpz}, \msg^{\xpo}\in\MSg$ and a set $\chYSet\subset\USet$ such that: $|\chYSet|\leq m$ and $\XSet\not\subset\chYSet$ for all $\XSet$ queried for secret key. $\AC$ returns $\Ct^{\xpb}\leftarrow \Enc(\mpk, \chYSet, \msg^{\xpb})$ for $\bee\sample\{0,1\}$.
    
    \item $\QPhase{II}$: Same as $\QPhase{I}$ subject to the restrictions: 
    \begin{itemize}
        \item Key query on $\XSet\subset \USet$ s.t. $|\XSet|\leq m$ is allowed if $\XSet\not\subset\chYSet$.
        \item $\Ct^{\xpb}$ is not given as a decryption query.
    \end{itemize}

    \item \textbf{Guess}: $\AA$ outputs its guess $\bee'\in\{0,1\}$ and wins if $\bee=\bee'$.
\end{itemize}
For any adversary $\AA$ the advantage is,
{\[\mathbf{Adv}_{\AA,\spe}^{\IndCCA}=|\Pr[\bee=\bee']-1/2|.\]}
\noindent A subset-predicate encryption scheme $\spe$ is said to be $\IndCCA$ secure if for all $\ppt$ adversary $\AA$, $\Adv{\AA,\ibr}{\IndCCA}\leq \neglgbl$. 
If there is an $\IniT$ phase before the $\SetuP$ where the adversary $\AA$ commits to the challenge data-attribute set $\chYSet$, we call such security model as selective $\IndCCA$ security $(\IndsCCA)$ model.

{The term \emph{Predicate Encryption} typically encompasses security notion such as attribute-hiding. Subset Predicate Encryption, originally coined by Katz \etal \cite{CANS:KMMS17} did not consider this security property. In this work, we focus on efficiency in terms of ciphertext-size and key-size, rather than attribute-hiding or key-hiding security. This work follows the standard practice of calling this primitive as Subset Predicate Encryption with an acceptance of the potential miscalling.}    

\section{Construction} \label{sec:Construct}
\def\CRHfnc{\mathsf{CH}}
\def\UHfnc{\mathsf{UH}}

We present an SPE construction $\spe$ with a constant-size secret key and a constant-size ciphertext in the composite order pairing setting. Let $\Ssk=\{0,1\}^{\mathsf{poly}({\secpp})}$ is the kem-key space. Let $\CRHfnc=\{f: \Ssk\times \sbG\times \sbG \rightarrow \ZN\}$ is the set of all collision-resistant hash functions and $\UHfnc=\{f: \sbGT\rightarrow \Ssk\}$ is the set of all universal hash functions. The following four algorithms define $\spe$ in \Cref{fig:Construction-SPE}.

\begin{figure*}[ht]
    \centering
    \fbox{
        \begin{minipage}{0.52\textwidth}
            \underline{$\Setup(1^{\secpp},m)$}
            \begin{algorithmic}[1]
                \State $(N,\sbG,\sbGT,e)\leftarrow \SBGgen(1^{\secpp})$ for $N=p_1p_2p_3$.
                \State $\hfnc_1\sample \CRHfnc$, $\hfnc_2\sample \UHfnc$.
                \State $\amsk,\bmsk,a,b\sample \ZN$, $\sbgo,u\sample \sbGo$, $\sbgr\sample\sbGr$. 
                \State For $i\in[m]$, 
                \item[] \hspace{0.2cm} $G_i=\sbgo^{\amsk^i}$, $U_i=u^{\amsk^i}\cdot R_{3,i}$ for $R_{3,i}\sample\sbGr$. 
                \State $\msk=(\amsk,\bmsk, u, \sbgr)$.
                \State $\mpk=(\sbgo,\sbgo^{\bmsk},\sbgo^a,\sbgo^b,e(\sbgo,u)^\bmsk,\sbgo^{\bmsk a},\sbgo^{\bmsk b},\hfnc_1,\hfnc_2,\left\{G_i, U_i\right\}_{i\in[m]})$
                \item[] 
                \item[] 
            \end{algorithmic}
            \underline{$\Enc(\mpk,\YSet, \msg)$}
            \begin{algorithmic}[1]
                \State Let $\YSet\subset \ID$ s.t. $|\YSet|=\Ysz\leq m$.
                \State Let $P_{\YSet}(\zee) =\Prod{y}{\YSet} (\zee+y)=\Sum{i}{[\Ysz]} b_i \zee^i$.
                \State $s\sample\ZN$. 
                \State $\Ct=(\hat{\ct}, \ct_0,\ct_1, \ct_2)$ where,
                \item[] \hspace{0.5cm}$\hat{\ct}=\msg\oplus\ssk$,
                \item[] \hspace{0.5cm}$\ct_0=\sbgo^{\bmsk s}$,
                \item[] \hspace{0.5cm}$\ct_1=\sbgo^{P_{\YSet}(\amsk) s}$,
                \item[] \hspace{0.5cm}$\ct_2= \sbgo^{\bmsk(a+\tau b)s}$,
                \item[] \hspace{-0.3cm}for $\ssk=\hfnc_2(e(\sbgo, u)^{\bmsk s})$ and $\tau=\hfnc_1(\hat{\ct},\ct_0,\ct_1)$.
            \end{algorithmic}        
            \end{minipage}
        \vline \hspace{1pt}
        \begin{minipage}{0.43\textwidth}
            \underline{$\Kgen(\msk, \XSet)$}
            \begin{algorithmic}[1]
                \State Let $\XSet\subset \ID$ s.t. $|\XSet|=\Xsz\leq m$.
                \State Let $P_{\XSet}(\zee) =\Prod{x}{\XSet} (\zee+x)=\Sum{i}{[\Xsz]} a_i \zee^i$.
                \State $\sk=u^{\frac{\bmsk}{P_{\XSet}(\amsk)}}\cdot X_3$ for $X_3\sample \sbGr$.
                \item[] 
                \item[]
                \item[] 
                \item[]
            \end{algorithmic}
        \underline{$\Dec(\mpk, (\XSet,\sk), (\YSet,\Ct))$}
            \begin{algorithmic}[1]
                \State Let $\Ct=(\hat{\ct}, \ct_0,\ct_1, \ct_2)$. 
                \State Let $\tau'=\hfnc_1(\hat{\ct},\ct_0,\ct_1)$. 
                \State $r\sample\ZN$.
                \State Consistency check: Abort with output $\bot$ if fails.\[e(\ct_2, \sbgo^r)\cdot e(\ct_1, \sbgo^{\bmsk r})\iseq e(\ct_0, (\sbgo^{a}\cdot (\sbgo^{b})^{\tau'}\cdot \sbgo^{P_{\YSet}(\amsk)})^r).\]
                \State $P_{\YSet\setminus\XSet}(\zee)=\Prod{v}{\YSet\setminus\XSet}(\zee+v)=\Sum{i}{[\Usz]} c_i \zee^i$.
                \State $\mu=e(\ct_1,\sk)\cdot e(\ct_0, \Prod{i}{[1,\Usz]} U_i^{c_i})^{-1}$.
                \State $\msg'=\hat{\ct}\oplus\hfnc_2(\mu^{1/{c_0}})$.
            \end{algorithmic}     
        \end{minipage}
    }
    \caption{Our SPE Construction: $\spe$}
    \label{fig:Construction-SPE}
\end{figure*}


\subsection{Correctness}\label{sec:SPE-Correctness}
\begin{itemize}
    \item $\tau'=\hfnc_1(\hat{\ct},\ct_0,\ct_1)=\tau$. 
    \item $\XSet\subseteq\YSet\Rightarrow \exists\ \USet\subset \ID\setminus\XSet$ s.t. $\USet\cup \XSet=\YSet$
\end{itemize}

$\star$ \textit{Consistency Check:}
\vspace{-0.3cm}
\begin{equation*}
    \begin{aligned}
        &e(\ct_0, \sbgo^{a}\left(\sbgo^{b}\right)^{\tau'}\sbgo^{P_{\YSet}(\amsk)}) = e(\sbgo^{\bmsk s}, \sbgo^{(a+\tau' b+P_{\YSet}(\amsk))})\\ =&\ e(\sbgo^{\bmsk(a+\tau b)s}, \sbgo)e(\sbgo^{P_{\YSet}(\amsk)s}, \sbgo^{\bmsk})=e(\ct_2, \sbgo)e(\ct_1, \sbgo^{\bmsk}).\\
    \end{aligned}
\end{equation*}

$\star$ \textit{CPA-Ciphertext Decryption:} $P_{\USet}(\zee)=\Prod{v}{\USet}(\zee+v)=\Sum{i}{[\Usz]} c_i \zee^i$.

\begin{equation*}
    \hspace{-0.5cm}
    \begin{aligned}
        A   &= e(\ct_0, \Prod{i}{[1,\Usz]} U_i^{c_i}) = e(\sbgo^{\bmsk s}, \Prod{i}{[1,\Usz]}\left(u^{\amsk^i}\cdot R_{3,i}\right)^{c_i})\\ 
            &= e(\sbgo, u)^{\bmsk s \Sum{i}{[1,\Usz]} c_i \amsk^i} \\
            &= e(\sbgo, u)^{\bmsk s (P_{\USet}(\amsk)-c_0)} \\
        B   &= e(\ct_1, \sk) = e(\sbgo^{P_{\YSet}(\amsk) s}, u^{\frac{\bmsk}{P_{\XSet}(\amsk)}}) = e(\sbgo,u)^{\bmsk s P_{\YSet\setminus \XSet}(\amsk)}\\ 
            &= e(\sbgo,u)^{\bmsk s P_{\USet}(\amsk)}\\
        (B\cdot A^{-1})^{1/{c_0}} &= \left(e(\sbgo,u)^{\bmsk s c_0}\right)^{1/{c_0}} = e(\sbgo,u)^{\bmsk s}\\
    \end{aligned}
\end{equation*}
Thus, $\hat{\ct}\oplus \hfnc_2((B\cdot A^{-1})^{1/{c_0}})=\msg$.

\subsection{Security} \label{sec:SPE-Sec}
\begin{theorem}\label{thm:SPE-Sec}
    If the $\dgho$ assumption and the $\dght$ assumption hold in $\abG$, $\hfnc_1$ is a collision-resistant hash function and $\hfnc_2$ is a universal hash function, then $\spe$ is an $\IndsCCA$ secure SPE scheme. Precisely, for any $\ppt$ adversary $\AA$ that breaks $\spe$ in the $\IndsCCA$ model making at most $q$ trapdoor queries and at most $\qd$ decryption requests, there exist $\ppt$ algorithms $\AB_1,\AB_2,\AB_3$ such that, 
    \begin{equation*}
        \begin{aligned}
            \Adv{\AA, \IndsCCA}{\spe}&\leq (2\qd+2)\Adv{\AB_1}{\dgho}\\
            &\quad+(m+q+2)\Adv{\AB_2}{\dght}+\qd\cdot \Adv{\AB_3}{\hfnc_1}.\\
        \end{aligned}
    \end{equation*}
\end{theorem}

\begin{proofsketch}
    The proof is established via a hybrid argument. 
    The idea is to modify each game only a small amount, which allows the solver $\AB$ to model the intermediate games properly.
    The hybrid argument is based on Wee's~\cite{TCC:Wee14} porting of $\dejaq$ framework introduced by Chase and Meiklejohn~\cite{EC:ChaMei14}.
    In the first game $\game{0}$, both the challenge ciphertext, {\color {black}consistency check} and secret keys are normal.
    $\game{1}$ differs from $\game{0}$ as we introduce some simplifying natural restrictions.
    $\game{2}$ reformulates the challenge ciphertext to a different representation that will be useful in later games.
    In $\game{3}$, we replace the challenge ciphertext component $\ct_0$ with a random $\sbGo$ element.
    Then, in $\game{4}$, we introduce a semi-functional component in the challenge ciphertext.
    We next define a sub-sequence of games $(\game{5,0,1},\game{5,1,0}, \game{5,1,1}, $ $\game{5,2,0},\game{5,2,1}, \ldots, \game{5,m+q+1,0},$ $\game{5,m+q+1,1})$ to introduce entropy into the semi-functional components of secret key and few related public parameters.
    We denote $\game{4}=\game{5,0,1}$ and $\game{5}=\game{5,m+q+1,1}$.
    Note that till this point, we mostly have followed the proof of~\cite{RSA:ChaMuk19} and from now on, we argue the proof differently.
    We show that the above sub-sequence of games effectively introduces enough entropy in the semi-functional component of $U_0$ such that we can replace it with pure random choice in $\game{6}$.
    This argument, unlike \cite{RSA:ChaMuk19}, allows any secret key extraction query provided it does not collide with the natural restriction.
    We next define a sub-sequence of games $(\game{7,0,3}, \game{7,1,1}, \game{7,1,2},$\\ $\game{7,1,3}, \ldots,\game{7,\qd,1},\game{7,\qd,2},\game{7,\qd,3})$ to introduce entropy into the \emph{consistency check}.
    We denote $\game{6}=\game{7,0,3}$ and $\game{7}=\game{7,\qd,3}$.
    Finally, in $\game{8}$, we show that semi-functional components as a whole supply enough entropy to hide  $\ssk$.
    Let us denote $\win{i}$ be the event that $\AA$ has won the game $\game{i}$.

    Note that, above, we have mentioned that till $\game{5}$, we have mostly followed the argument of \cite{RSA:ChaMuk19}.
    To detail the novelty of our work, we recall their technique to point out the difference.
    In particular, in $\game{6}$, Chatterjee and Mukherjee \cite{RSA:ChaMuk19} replaced all the semi-functional exponents with purely random variables. 
    We explain their four-step proof technique abusing the notations a bit.
    \begin{enumerate}
        \item To represent the semi-functional components of the queried secret keys $\sk_{\XSet_1}$, $\ldots$, $\sk_{\XSet_{q}}$ in terms for a linear system of equations where $\hat{\D}$ denotes the linear transformations.
        \item To show the matrix $\hat{\D}$ in \Cref{eq:hat-D} is non-singular where $Q=|\cup_{i\in[q]} X_i|$. 
        \item To show the matrix $\hat{\D}'$ 
        in \Cref{eq:hat-D-dash} can be achieved via elementary row operations on $\hat{\D}$. The goal is to show $\hat{\D}'$ is also full-rank.
        \item Finally, to show the matrix $\hat{\D}''$ in \Cref{eq:hat-D-dash-dash}
        is a full-rank matrix.
    \end{enumerate}
    \begin{figure}
    \begin{scriptsize}
        \vspace{-0.5cm}
        \begin{multicols}{2}
            \begin{equation}
                \hat{\D}=\left[
                    \begin{matrix}
                        \frac{1}{(\zee_1+x_1)} &  \ldots & \frac{1}{(\zee_1+x_Q)}\\
                        \vdots & \ddots & \vdots \\
                        \frac{1}{(\zee_Q+x_1)} &  \ldots & \frac{1}{(\zee_Q+x_Q)}\\
                    \end{matrix}
                \right] \label{eq:hat-D},
            \end{equation}
            \break 
            \begin{equation}
                \hspace{-0.2cm}
                \hat{\D}'=\left[
                    \begin{matrix}
                        \frac{1}{(P_{\XSet_1}(\zee_1))} &  \ldots & \frac{1}{(P_{\XSet_q}(\zee_1))}\\
                        \vdots & \ddots & \vdots \\
                        \frac{1}{(P_{\XSet_q}(\zee_Q))} &  \ldots & \frac{1}{(P_{\XSet_q}(\zee_Q))}\\
                    \end{matrix}
                \right] \label{eq:hat-D-dash}, 
            \end{equation}
            \break 
            \begin{equation}
                \hat{\D}''=\left[
                    \begin{matrix}
                        1 & \ldots & 1\\
                        \frac{1}{(P_{\XSet_1}(\zee_1))} &  \ldots & \frac{1}{(P_{\XSet_q}(\zee_1))}\\
                        \vdots & \ddots & \vdots \\
                        \frac{1}{(P_{\XSet_q}(\zee_Q))} &  \ldots & \frac{1}{(P_{\XSet_q}(\zee_Q))}\\
                    \end{matrix}
                \right] \label{eq:hat-D-dash-dash}
            \end{equation}
          \end{multicols}
    \end{scriptsize}        
    \end{figure}

    Their fourth step extensively utilizes the restrictions imposed by their so-called \emph{selective$^*$ security model}.

    We, however, notice that the fourth step of the argument of \cite{RSA:ChaMuk19} is sufficient but not necessary.
    In particular, for the argument to hold, $\hat{\D}'$ need not be non-singular.
    Looking ahead, it is sufficient for the security argument to show that a particular row of $\hat{\D}''$ is independent of the rest.
    Since this is a more relaxed requirement on $\hat{\D}''$ (and subsequently on $\hat{\D}'$ and $\hat{\D}$), we could allow even \emph{clawed} query sequence.
    Therefore, we could remove the \emph{cover-free set} restrictions imposed in \cite{RSA:ChaMuk19}, which was completely unrealistic for any practical consideration. 


We revisit from the notion of \emph{cover-free sets} given by \cite{RSA:ChaMuk19} to present a justification that this restriction is unnecessary.
As has been pointed out, such a restrictive query sequence was introduced to stop queries like $\XSet_1=\{a,b\},\XSet_2=\{b,c\},\XSet_3=\{a,c\}$.
For $\XSet'=\{a,b,c\}$, it is easy to see that $\sk_{\XSet'}$ can be computed via two paths here: $(i)$ pair $(\sk_{\XSet_1},\sk_{\XSet_2})$, $(ii)$ pair $(\sk_{\XSet_2},\sk_{\XSet_3})$ where $sk_{\XSet}$ is defined using $\frac{1}{P_{\XSet}(\zee)}=\frac{1}{\Prod{x}{\XSet} (\zee+x)}$.
As was discussed in Introduction, this follows from the fact that $\frac{1}{(\zee+a)(\zee+b)(\zee+c)} = \frac{1}{c-a}\cdot\frac{1}{(\zee+a)(\zee+b)}+\frac{1}{a-c}\cdot\frac{1}{(\zee+b)(\zee+c)}=\frac{1}{b-a}\cdot\frac{1}{(\zee+a)(\zee+c)}+\frac{1}{a-b}\cdot\frac{1}{(\zee+b)(\zee+c)}$. 
As \cite{RSA:ChaMuk19} wanted to argue that all the rows of $\hat{\D}'$ are independent, such a query sequence imposes a roadblock.
We, on the other hand, argue that a particular row of $\hat{\D}'$ is independent of $\sk_{\XSet_1},\sk_{\XSet_2},\sk_{\XSet_3}$.
More precisely, we argue that $(11\ldots 1)$ is independent of $\sk_{\XSet_1},\sk_{\XSet_2},\sk_{\XSet_3}$ where $sk_{\XSet_i}$ is defined using $\frac{1}{P_{\XSet_i}(\zee)}=\frac{1}{\Prod{x}{\XSet_i} (\zee+x)}$.
This simple yet critical observation allowed us to drop the restrictions on the queried sets.
However, this introduces some challenges on its own as we allow significantly more key extraction queries with overlapping sets.
We formulate the argument completely next.


\end{proofsketch}

\begin{proof}[Proof of \Cref{thm:SPE-Sec}] We define the hybrid games next.

    \begin{description}
        \item[$\game{0}$.] This is same as the real game.
    
        \item[$\game{1}$.] The following natural assumptions are made on the game.
        \begin{itemize}
            \item For all $z\in (\YSet\cup \bigcup_{i\in[q]} \XSet_i)$, $(\amsk+z)$ is not divisible by $\prm_1$. Otherwise, $\AB$ can easily solve the subgroup decision assumption $\dgho$ by computing $\mathsf{gcd}((\amsk+z), N)$.
            
            \item For all distinct $i,j\in[q]$, for all $x, x'\in \XSet_{i,j}$, if $x\neq x \mod{N}$ then $x\neq x' \mod{p_2}$. Otherwise, $\AB$ can easily solve the subgroup decision assumption $\dght$ by computing $\mathsf{gcd}((x-x'), N)$.
        \end{itemize}
        \vspace{.1cm}
        Therefore, $|\prob{\win{1}}-\prob{\win{0}}|\leq \Adv{{\AB_1}}{\dgho} + \Adv{{\AB_2}}{\dght}$. 

        \item[$\game{2}$.] {A conceptual change to $\game{1}$ is performed here. Given the challenge $\YSet=\{y_1,\ldots,y_{\Ysz}\}$, 
        pick $\amsk, \tbmsk,a,b, u\sample \ZN^4\times \sbGo$. 
        Define the polynomial $P_{\YSet}(\zee)=\Prod{y}{\YSet}(\zee+y)$ and set $\bmsk={\tbmsk}\cdot P_{\YSet}(\amsk)\mod{N}$.
        This affects only $\sbgo^{\bmsk},\sbgo^{\bmsk a},\sbgo^{\bmsk b}$ in $\mpk$.
        The rest of the public parameters in $\mpk$ are defined the same as $\game{1}$.
        The secret keys corresponding to $\XSet_i$ is $\sk_i=u^{\frac{{\tbmsk}\cdot P_{\YSet}(\amsk)}{P_{\XSet_i}(\amsk)}}\cdot X_{3,i}$ for $i\in[q]$.
        The ciphertext is \[\hat{\ct}=\msg^{\xpb}\oplus\ssk, \ct_0=\sbgo^{s{\tbmsk}P_{\YSet}(\amsk)}, \ct_1=\ct_0^{1/{\tbmsk}},\ct_2 =\ct_0^{(a+\tau b)}\] 
        where $\ssk=\hfnc_2(e(\ct_0, U_0))$, $\tau=\hfnc_1(\hat{\ct},\ct_0,\ct_1)$ and $U_0=u\cdot R_3$ for $R_3\sample \sbGr$.
        Note that, even after replacing $\bmsk$ with ${\tbmsk}\cdot P_{\YSet}(\amsk)\mod{N}$ the ciphertext distribution stays the same as ${\tbmsk}$ is uniformly random and $P_{\YSet}(\amsk)\neq 0 \mod{\prm_1}$.
        Therefore, $\prob{\win{2}}=\prob{\win{1}}$.}

        \item[$\game{3}$.] Here again, a conceptual change is performed to $\game{2}$. 
        Choose $\ct_0\sample \sbGo$.
        The rest of the ciphertext is defined as in $\game{2}$. 
        As all of $\ssk$, $\ct_1$ and $\ct_2$ are functions of $\ct_0$, namely $\ssk=\hfnc(e(\ct_0, U_0))$, $\ct_1=\ct_0^{1/{\tbmsk}}$ and $\ct_2 =\ct_0^{(a+\tau b)}$ where $\tau=\hfnc_1(\hat{\ct},\ct_0,\ct_1)$, such a replacement does not change the distribution of ciphertext.
        Therefore, $\prob{\win{3}} = \prob{\win{2}}$.

        \item[$\game{4}$.] Here, 
        $\ct_0$ is chosen from the group $\sbGot$ uniformly at random. 
        The rest of the ciphertext and secret keys are generated similarly to $\game{3}$. 
        We argue that such a change is ``invisible'' to any $\ppt$ adversary due to the subgroup decision assumption $(\dgho)$ i.e., $\abslt{\prob{\win{4}}-\prob{\win{3}}}\leq \Adv{\AB}{\dgho}$.
        We describe a reduction informally to provide a proof sketch here.
        Given a $\dgho$ problem instance, 
        $\AB$ chooses $\amsk,{\tbmsk},a,b\sample \ZN$ to compute $\mpk$ similar to $\game{3}$.
        The challenge phase uses the target $T$ of $\dgho$ problem instance to simulate $\ct_0$.
        If $T$ was from $\sbGo$, the $\ct_0$ is normal whereas if $T$ was from $\sbGot$, the $\ct_0$ is semi-functional.
        Since $\ct_0$ determines the challenge ciphertext completely, the distribution from which $T$ was chosen determines if the challenge ciphertext is normal or semi-functional.

        \item[$\game{5}$.] 
        {Here, the secret keys $\sk_i$ for all $i\in[q]$ are changed gradually to make them semi-functional. 
        In the process, $U_i$ in $\mpk$ for all $i~\in~[m]$ are also modified gradually. 
        Precisely, we change the public parameter $U_i$ from 
        \[u^{\amsk^i}\cdot R_{3,i} \text{\ \ \ to \ \ \ }u^{\amsk^i}\cdot \sbgt^{\sum_{j\in[\colcnt]} \kyr[j]\amsk_j^i} \cdot R'_{3,i}\] 
        and the secret key $\sk_i$ for each queried set $\XSet_i$ is changed from 
        \[\hspace{-0.4cm}u^{\frac{\tbmsk\cdot P_{\YSet}(\amsk)}{P_{\XSet_i}(\amsk)}}\cdot X_{3,i} \text{\ \ to \ \ }u^{\frac{\tbmsk\cdot P_{\YSet}(\amsk)}{P_{\XSet_i}(\amsk)}}\cdot\sbgt^{\sum_{j\in[\colcnt]} \frac{\kyr[j] \cdot\tbmsk\cdot P_{\YSet}(\amsk_j)}{P_{\XSet_i}(\amsk_j)}} \cdot X'_{3,i}\] 
        for $\kyr[1],\ldots, \kyr[\colcnt], \amsk_1,\ldots,$ $\amsk_{\colcnt}\sample \ZN$.
        This is done via sequence of intermediate games namely $\game{{5,1,0}}$, $\game{{5,1,1}}$, $\ldots$,$\game{{5,\colcnt,0}}$, $\game{{5,\colcnt,1}}$. 
        We denote $\game{4}$ by $\game{{5,0,1}}$ and $\game{5}$ by $\game{{5,\rowcnto+\rowcntt+1,1}}$.

        \begin{itemize}
            \item {In $\game{{5,k,0}} (k\in[\colcnt])$, the public parameter $U_i$ for $i\in[m]$ is changed as follows. 
            We also similarly change $U_0$.
            \[
            u^{\amsk^i}\cdot \sbgt^{\sum_{j\in[k-1]} \kyr[j]\amsk_j^i} R'_{3,i} \rightarrow u^{\amsk^i}\cdot \boxed{\sbgt^{\kyr\amsk^i}}\cdot \sbgt^{\sum_{j\in[k-1]} \kyr[j]\amsk_j^i} R'_{3,i}.\]
            For all $i\in[q]$, the $i^{th}$ secret key $\sk_i$ is changed as follows.
            \begin{multline}\label{eqn:SPE-I-key-i}
                u^{\frac{\tbmsk\cdot P_{\YSet}(\amsk)}{P_{\XSet_i}(\amsk)}}\cdot\sbgt^{\sum_{j\in[k-1]} \frac{\kyr[j] \cdot\tbmsk\cdot P_{\YSet}(\amsk_j)}{P_{\XSet_i}(\amsk_j)}} X'_{3,i} \\
                \rightarrow u^{\frac{\tbmsk\cdot    P_{\YSet}(\amsk)}{P_{\XSet_i}(\amsk)}} \cdot\boxed{\sbgt^{\frac{\kyr\cdot\tbmsk\cdot P_{\YSet}(\amsk)}{P_{\XSet_i}(\amsk)}}}\cdot
                    \sbgt^{\sum_{j\in[k-1]} \frac{\kyr[j] \cdot\tbmsk\cdot P_{\YSet}(\amsk_j)}{P_{\XSet_i}(\amsk_j)}} X'_{3,i}.
            \end{multline}
       
            \item  {In $\game{{5,k,1}} (k\in[0,\colcnt])$, the parameters $\{U_i\}_{i\in[0,m]}$ 
            and secret key $\{\sk_i\}_{i\in[q]}$ distributions are respectively given by,
            \begin{equation*}
                \begin{aligned}
                    U_i&=u^{\amsk^i}\cdot \sbgt^{\sum_{j\in[k]} \kyr[j]\amsk_j^i} R'_{3,i} \\
                    \sk_i&=u^{\frac{\tbmsk\cdot P_{\YSet}(\amsk)}{P_{\XSet_i}(\amsk)}}\cdot \sbgt^{\sum_{j\in[k]} \frac{\kyr[j] \cdot\tbmsk\cdot P_{\YSet}(\amsk_j)}{P_{\XSet_i}(\amsk_j)}} X'_{3,i}. \\
                \end{aligned}
            \end{equation*}
            }   
        }
        \end{itemize}
        Here, we argue that $\Adv{\AA}{\game{{5,k,0}}} = \Adv{\AA}{\game{{5,k,1}}}$.
        Firstly, the public parameters except $\{G_i,U_i\}_{i\in[m]}$ do not involve $\amsk$.
        Then, $\ \amsk\!\!\mod{p_1}$ in $\{G_i,U_i\}_{i\in[m]}$ does not leak any information regarding $\ \amsk\!\!\mod{p_2}$ due to \emph{Chinese Remainder Theorem}. 
        Finally, $\ct_0$ is chosen uniformly at random from $\sbGot$, it is completely independent of $\ \amsk\!\!\mod{p_2}$.
        Thus, the changes between $\game{{5,k,0}}$ and $\game{{5,k,1}}$ is invisible to any (unbounded) adversary $\AA$.

        \vspace{.1cm}
        \Cref{lem:SPE-I-k1-k0} argues that $\game{{5,k-1,1}}$ and $\game{{5,k,0}}$ are indistinguishable for any $k\in[\colcnt]$ under the $\dght$ assumption.
        This concludes that $|\prob{\win{5}}-\prob{\win{4}}|\leq (\colcnt)\cdot \Adv{\AB}{\dght}$.

        \begin{lemma}\label{lem:SPE-I-k1-k0}{There exists a $\ppt$ adversary $\AB$ such that, \[|\prob{\win{{5,k-1,1}}}-\prob{\win{{5,k,0}}}|\leq \Adv{\AB}{\dght}.\]}
        \end{lemma}

        We reproduce this proof for completeness in \Cref{sec:Deferred}.
    }

    \item[$\game{6}$.] In this game, we simulate $U_0=u\cdot g_2^{z_0}$ where $z_0\sample\Zp[2]$. 
    We argue that, $\Adv{\AA}{\game{{6}}} = \Adv{\AA}{\game{{5}}}$ for any (unbounded) adversary $\AA$. In the following discussion, by \emph{semi-functional exponent} of any group element, we denote the exponent present in the $\sbGt$ subgroup. For example, the semi-functional exponent of $U_0$ in $\game{6}$ is $z_0$. We also look at semi-functional exponents of $\{U_i\}_{i\in[m]}$ and $\{\sk_j\}_{j\in[q]}$. Assume that $z_1,\ldots,z_m$ denotes the semi-function exponents of $U_1,\ldots,U_m$ respectively and $z_{m+1},\ldots,$ $z_{m+q}$ denotes the semi-function exponents of $\sk_1,\ldots,\sk_q$ respectively. We gather these semi-functional exponents as a linear system of equations $(\z=\A\r)$ in \Cref{eqn:SPE-I-matrix}. To show that $z_0$ is uniformly random and independent of $z_1,\ldots,z_{m+q}$, it suffices to show that the first row of $\A$ is independent of the rest of the rows of $\A$. This is because, if $\c^\top\notin\Span(\M)$ then $\c^\top\r$ is completely independent of $\M\r$.
    Let us assume $\A=\iCol{\Au}{\Ad}$ where $\Au$ is the first row of $\A$, and $\Ad$ is a matrix that contains rest of the rows of $\A$.

    \begin{equation}\label{eqn:SPE-I-matrix}
        \hspace{0.6cm}
        \underbrace{\left[\begin{matrix}z_0\\ z_1\\\vdots\\ z_\rowcnto\\ z_{\rowcnto+1}\\\vdots \\z_{\rowcnto+\rowcntt}\end{matrix}\right]}_{\z} = 
    \underbrace{\left[\begin{matrix} 
        1 & 1 & \ldots & 1 \\
        \amsk_1 & \amsk_2 & \ldots & \amsk_{m+q+1}\\
        \amsk_1^2 & \amsk_2^2 & \ldots & \amsk_{m+q+1}^2\\
        \vdots & \vdots & \ddots & \vdots \\
        \amsk_1^m & \amsk_2^m & \ldots & \amsk_{m+q+1}^m\\
        \frac{\tbmsk P_{\YSet}(\amsk_1)}{P_{\XSet_1}(\amsk_1)} & \frac{\tbmsk P_{\YSet}(\amsk_2)}{P_{\XSet_1}(\amsk_2)} & \ldots & \frac{\tbmsk P_{\YSet}(\amsk_{m+q+1})}{P_{\XSet_1}(\amsk_{m+q+1})}\\
        \frac{\tbmsk P_{\YSet}(\amsk_1)}{P_{\XSet_2}(\amsk_1)} & \frac{\tbmsk P_{\YSet}(\amsk_2)}{P_{\XSet_2}(\amsk_2)} & \ldots & \frac{\tbmsk P_{\YSet}(\amsk_{m+q+1})}{P_{\XSet_2}(\amsk_{m+q+1})}\\
        \vdots & \vdots & \ddots & \vdots \\
        \frac{\tbmsk P_{\YSet}(\amsk_1)}{P_{\XSet_q}(\amsk_1)} & \frac{\tbmsk P_{\YSet}(\amsk_2)}{P_{\XSet_q}(\amsk_2)} & \ldots & \frac{\tbmsk P_{\YSet}(\amsk_{m+q+1})}{P_{\XSet_q}(\amsk_{m+q+1})}\\
    \end{matrix}\right]}_{\A}\cdot
    \underbrace{\left[\begin{matrix}\kyr[1]\\ \kyr[2]\\\vdots\\ \kyr[\colcnt]\end{matrix}\right]}_{\r}.
    \end{equation}

    \begin{lemma}\label{lem:SPE-matrix-nonsingular}
        $\Au\notin \Span(\Ad)$.
    \end{lemma}
        \begin{proof}
            We denote $\Ad=\iCol{\B}{\P}$ where $\B\in\Zp[2]^{\rowcnto\times (\colcnt)}$ is the first $m$ rows of $\Ad$ and $\P\in\Zp[2]^{\rowcntt\times (\colcnt)}$ is last $\rowcntt$ rows of $\Ad$.

            Recall that, for all $i^{th}$ query, $\XSet_i\not\subset \YSet$ due to natural restriction.
            Let us denote $P_{\YSet}(\zee)=\Prod{y}{\YSet} (\zee+y)=\Sum{i}{[0,\Ysz]} C_i \zee^i$ for $C_{\Ysz}=1$.
            Then, 

            \begin{equation*}
                \scriptsize
                \hspace{-0.3cm}\begin{aligned}
                    \frac{P_{\YSet}(\zee)}{P_{\XSet_{\ii}}(\zee)}
                    &= \frac{\Sum{i}{[0,\Ysz]} C_i \zee^i}{\Prod{x_k}{\XSet_{\ii}} (\zee+x_k)} =\Sum{x_k}{X_{\ii}} \frac{1}{\Prodneq{x_{\jj}}{X_{\ii}}{x_{\jj}}{x_k} (x_{\jj}-x_k)}\cdot\frac{\Sum{i}{[0,\Ysz]} C_i \zee^i}{\zee+x_k}\\
                    &=\Sum{x_k}{X_{\ii}} \frac{1}{\Prodneq{x_{\jj}}{X_{\ii}}{x_{\jj}}{x_k} (x_{\jj}-x_k)}\cdot\Bigg(\Summ{j}{[0,\Ysz-i-1]}{i}{[0,\Ysz-1]} (-x_k)^j C_{i+j+1} \zee^i  
                    + \frac{P_{\YSet}(-x_k)}{\zee+x_k}\Bigg)\\
                    &= D_{\YSet,\XSet_{\ii}}(\zee)+V_{\YSet,\XSet_{\ii}}(\zee)\\
                \end{aligned}
            \end{equation*}
            where $C_i,x_{\jj},x_k$ are all non-trivial scalars except with negligible probability, $D_{\YSet,\XSet_{\ii}}(\zee)=$
            \begin{equation*}
                \begin{aligned}
                    \Summ{i}{[0,\Ysz-1]}{x_k}{X_{\ii}} \frac{1}{\Prodneq{x_{\jj}}{X_{\ii}}{x_{\jj}}{x_k} (x_{\jj}-x_k)}\Sum{j}{[0,\Ysz-i-1]} (-x_k)^j C_{i+j+1} \zee^i \\
        \end{aligned}
    \end{equation*}
            and $V_{\YSet,\XSet_i}(\zee)= \Sum{x_k}{X_{\ii}} \frac{1}{\Prodneq{x_{\jj}}{X_{\ii}}{x_{\jj}}{x_k} (x_{\jj}-x_k)}\cdot \frac{P_{\YSet}(-x_k)}{\zee+x_k}$.
            Observe that 
            \begin{equation*}
                \hspace{-0.3cm}
                \begin{aligned}
                    D_{\YSet,\XSet_{\ii}}(\zee)&=\Summ{i}{[0,\Ysz-1]}{x_k}{X_{\ii}} \frac{1}{\Prodneq{x_{\jj}}{X_{\ii}}{x_{\jj}}{x_k} (x_{\jj}-x_k)}\Sum{j}{[0,\Ysz-i-1]} (-x_k)^j C_{i+j+1} \zee^i \\
                    &= \Sum{i}{[0,\Ysz-1]} d_i^{(\ii)} \zee^i \\
                \end{aligned}
            \end{equation*}
            where $d_i^{(\ii)}=\Sum{x_k}{X_{\ii}} \frac{1}{\Prodneq{x_{\jj}}{X_{\ii}}{x_{\jj}}{x_k} (x_{\jj}-x_k)}\Sum{j}{[0,\Ysz-i-1]} (-x_k)^j C_{i+j+1}$ is a non-trivial scalar.
            Also,
            \begin{equation*}
                \begin{aligned}
                    V_{\YSet,\XSet_{\ii}}(\zee)
                    &=\Sum{x_k}{X_{\ii}} \frac{P_{\YSet}(-x_k)}{\Prodneq{x_{\jj}}{X_{\ii}}{x_{\jj}}{x_k} (x_{\jj}-x_k)}\cdot \frac{1}{\zee+x_k}\\ 
                    &= \Sum{x_k}{X_{\ii}} v_k^{(\ii)}\cdot \frac{1}{\zee+x_k}\\
                \end{aligned}
            \end{equation*}
            where $v_k^{(\ii)}=\frac{P_{\YSet}(-x_k)}{\Prodneq{x_{\jj}}{X_{\ii}}{x_{\jj}}{x_k} (x_{\jj}-x_k)}$ is a non-trivial scalar.
            Then, we express $\A=\F\cdot \D$ where,
            {\small
            \begin{equation}\label{eqn:matx-F}
                \F=
                \left[
                    \begin{matrix}
                        1        & 0       & 0       &\ldots& 0       & 0       & 0       &\ldots& 0\\
                        0        & 1       & 0       &\ldots& 0       & 0       & 0       &\ldots& 0\\
                        0        & 0       & 1       &\ldots& 0       & 0       & 0       &\ldots& 0\\
                        \vdots   &\vdots   &\vdots   &\ddots&\vdots   &\vdots   &\vdots   &\ddots&\vdots\\
                        0        & 0       & 0       &\ldots& 1       & 0       & 0       &\ldots& 0\\
                        d_0^{(1)}&d_1^{(1)}&d_2^{(1)}&\ldots&d_m^{(1)}&v_1^{(1)}&v_2^{(1)}&\ldots&v_Q^{(1)}\\
                        d_0^{(2)}&d_1^{(2)}&d_2^{(2)}&\ldots&d_m^{(2)}&v_1^{(2)}&v_2^{(2)}&\ldots&v_Q^{(2)}\\
                        \vdots   &\vdots   &\vdots   &\ddots&\vdots   &\vdots   &\vdots   &\ddots&\vdots\\
                        d_0^{(q)}&d_1^{(q)}&d_2^{(q)}&\ldots&d_m^{(q)}&v_1^{(q)}&v_2^{(q)}&\ldots&v_Q^{(q)}\\
                    \end{matrix}
                \right]
            \end{equation}}
            and 
            \begin{equation}\label{eqn:matx-D}
            \D = 
            \left[
                \begin{matrix}
                1 & 1 & \ldots & 1 \\
                \amsk_1 & \amsk_2 & \ldots & \amsk_{m+Q+1}\\
                \amsk_1^2 & \amsk_2^2 & \ldots & \amsk_{m+Q+1}^2\\
                \vdots & \vdots & \ddots & \vdots \\
                \amsk_1^m & \amsk_2^m & \ldots & \amsk_{m+Q+1}^m\\
                \frac{1}{(\amsk_1+x_1)} & \frac{1}{(\amsk_2+x_1)} & \ldots & \frac{1}{(\amsk_{m+Q+1}+x_1)}\\
                \frac{1}{(\amsk_1+x_2)} & \frac{1}{(\amsk_2+x_2)} & \ldots & \frac{1}{(\amsk_{m+Q+1}+x_2)}\\
                \vdots & \vdots & \ddots & \vdots \\
                \frac{1}{(\amsk_1+x_Q)} & \frac{1}{(\amsk_2+x_Q)} & \ldots & \frac{1}{(\amsk_{m+Q+1}+x_Q)}\\
            \end{matrix}\right]
            \end{equation}
            {where $\Xsz-1\leq m-1$ out of $(v_1^{(i)},v_2^{(i)},\ldots,v_Q^{(i)})$ are non-zero for all $i\in[q]$.}
            Now, if $(i)$ any query triplet $\XSet_{i_1},\XSet_{i_2},\XSet_{i_3}$ result in a \emph{claw} i.e. $\frac{1}{P_{\XSet_{i_3}}(\zee)}=\nu\cdot \frac{1}{P_{\XSet_{i_1}}(\zee)}+\rho\cdot \frac{1}{P_{\XSet_{i_2}}(\zee)}$, then $\frac{P_{\YSet}(\zee)}{P_{\XSet_{i_3}}(\zee)}=\nu\cdot \frac{P_{\YSet}(\zee)}{P_{\XSet_{i_1}}(\zee)}+\rho\cdot \frac{P_{\YSet}(\zee)}{P_{\XSet_{i_2}}(\zee)}$ or $(ii)$ any query is repeated i.e. $\frac{1}{P_{\XSet_{i_4}}(\zee)}=\frac{1}{P_{\XSet_{i_5}}(\zee)}$.
            This translates to 
            {\small
            \begin{equation*}
                \hspace{-0.0cm}
                \begin{matrix}
                    \Big(d_0^{(i_3)}&d_1^{(i_3)}&d_2^{(i_3)}&\ldots&d_m^{(i_3)}&v_1^{(i_3)}&v_2^{(i_3)}&\ldots&v_Q^{(i_3)}\Big)&\\    
                    =\nu\cdot (d_0^{(i_1)}&d_1^{(i_1)}&d_2^{(i_1)}&\ldots&d_m^{(i_1)}&v_1^{(i_1)}&v_2^{(i_1)}&\ldots&v_Q^{(i_1)})&\\
                    +\rho\cdot (d_0^{(i_2)}&d_1^{(i_2)}&d_2^{(i_2)}&\ldots&d_m^{(i_2)}&v_1^{(i_2)}&v_2^{(i_2)}&\ldots&v_Q^{(i_2)}).&\\
                \end{matrix}
            \end{equation*}}
            To argue $\Au\notin\Span(\Ad)$, it suffices to argue that $\Au\notin \Span(\Ad')$ where $\A'=\iCol{\Au}{\Ad}=\F'\cdot\D$ where $\F'$ is same as $\F$ except with the rows that are directly computed (i.e., repetition or results in a claw) from last $q$ rows of $\F$.
            We can run this procedure for all secret key queries and find such a ${\F'}$ that removes all secret key queries on dependent sets. Thus,
            {\small\begin{equation}\label{eqn:matx-FDash}
                \F'=
                \left[
                    \begin{matrix}
                        1        & 0       & 0       &\ldots& 0       & 0       & 0       &\ldots& 0\\
                        0        & 1       & 0       &\ldots& 0       & 0       & 0       &\ldots& 0\\
                        0        & 0       & 1       &\ldots& 0       & 0       & 0       &\ldots& 0\\
                        \vdots   &\vdots   &\vdots   &\ddots&\vdots   &\vdots   &\vdots   &\ddots&\vdots\\
                        0        & 0       & 0       &\ldots& 1       & 0       & 0       &\ldots& 0\\
                        d_0^{(1)}&d_1^{(1)}&d_2^{(1)}&\ldots&d_m^{(1)}&v_1^{(1)}&v_2^{(1)}&\ldots&v_Q^{(1)}\\
                        d_0^{(2)}&d_1^{(2)}&d_2^{(2)}&\ldots&d_m^{(2)}&v_1^{(2)}&v_2^{(2)}&\ldots&v_Q^{(2)}\\
                        \vdots   &\vdots   &\vdots   &\ddots&\vdots   &\vdots   &\vdots   &\ddots&\vdots\\
                        d_0^{(q')}&d_1^{(q')}&d_2^{(q')}&\ldots&d_m^{(q')}&v_1^{(q')}&v_2^{(q')}&\ldots&v_Q^{(q')}\\
                    \end{matrix}
                \right]
            \end{equation}}
            where $q'\leq q$ is the number of queries that did not result in claw.
            Now we show that $\left(\begin{matrix}0&0&\ldots&0\end{matrix}\right)$ is not in the span of $\left(\begin{matrix}v_1^{(i)}&v_2^{(i)}&\ldots&v_Q^{(i)}\end{matrix}\right)$ for all $i\in[q]$.
            Firstly, for all $i\in[q]$, $\left(\begin{matrix}v_1^{(i)}&v_2^{(i)}&\ldots&v_Q^{(i)}\end{matrix}\right)\neq \left(\begin{matrix}0&0&\ldots&0\end{matrix}\right)$ as that would mean semifunctional exponents of $\sk_i$ is $0$.
            Then, for all distinct $i,j\in[q]$, \[\left(\begin{matrix}v_1^{(i)}&v_2^{(i)}&\ldots&v_Q^{(i)}\end{matrix}\right)\neq \left(\begin{matrix}v_1^{(j)}&v_2^{(j)}&\ldots&v_Q^{(j)}\end{matrix}\right)\] as repetition has been removed in the previous step.
            Thus, $\left(\begin{matrix}0&0&\ldots&0\end{matrix}\right)$ not in the span of $\left(\begin{matrix}v_1^{(i)}&v_2^{(i)}&\ldots&v_Q^{(i)}\end{matrix}\right)$ for all $i\in[q]$.
            This in turn argues that \[\Fu=\left(\begin{matrix}1&0&0&\ldots&0&0&0&\ldots&0\end{matrix}\right)\] not in the span of \[\hspace{0.0cm}\left(\begin{matrix}d_0^{(i)}&d_1^{(i)}&d_2^{(i)}&\ldots&d_m^{(i)}&v_1^{(i)}&v_2^{(i)}&\ldots&v_Q^{(i)}\end{matrix}\right)\] for all $i\in[q]$.
            Moreover, $\Fu$ is not in the span of \[\left(\begin{matrix}0&b_1&b_2&\ldots&b_m&s_1&s_2&\ldots&s_Q\end{matrix}\right)\] for all $i\in[m]$.
            Thus, $\Fu$ is independent of all the other rows in $\F'$.

            Due to \Cref{lem:SPE-I-GLR18}, $\D$ is non-singular. Then, $\Fu\cdot\D$ is independent of $\F'\cdot\D$.
            Thus, $\Ad$ is independent of $\Au$.
            Thus, $z_0$ is independent of $(z_1,\ldots,z_{m+q})$ and randomly chosen $r_1,\ldots,r_{m+q+1}$ ensures that $z_0$ is uniformly random as well.
        \end{proof}

    \item[$\game{7}$.] Here, the decryption consistency checks are gradually made semi-functional type-3. To do that, we define the following sequence of security games and argue their indistinguishability: $(\game{7,0,3}, \game{7,1,1}, \game{7,1,2},\game{7,1,3},$ 
    
    $\ldots,\game{7,\qd,1},\game{7,\qd,2},\game{7,\qd,3})$.
    We denote the following too: $\game{6}=\game{7,0,3}$ and $\game{7}=\game{7,\qd,3}$.

    \begin{itemize}
    \item[$\game{7,i,1}$.] Here, the decryption consistency check of $i^{th}$ decryption query is modified.
    On $j^{th}$ decryption query on $(\XSet,\Ct,\YSet)$, 

    {
    \begin{itemize}
        \item $j<i$: It samples $r\sample\ZN$ to define the consistency check to be: Abort if \[\hspace{0cm}e(\ct_2, \sbgo^r)\cdot e(\ct_1, \sbgo^{\bmsk r})\neq e(\ct_0, \sbgo^{(a+\tau'b+P_{\YSet}(\amsk))r}\sbgt^{\delta})\] where $\tau'=\hfnc_1(\hat{\ct},\ct_0,\ct_1)$ and $\delta\sample\ZN$ chosen during $\Setup$.
        \item $j=i$: It samples $\sbgot\sample\sbGot$ to define the consistency check to be:
        Abort if \[\hspace{0cm}e(\ct_2, \sbgot)\cdot e(\ct_1, \sbgot^{\bmsk}) \neq e(\ct_0, \sbgot^{(a+\tau'b+P_{\YSet}(\amsk))})\] where $\tau'=\hfnc_1(\hat{\ct},\ct_0,\ct_1)$.
        \item $j>i$: It samples $r\sample\ZN$ to define the consistency check to be: 
        Abort if \[\hspace{0cm}e(\ct_2, \sbgo^r)\cdot e(\ct_1, \sbgo^{\bmsk r})\neq e(\ct_0, \sbgo^{(a+\tau'b+P_{\YSet}(\amsk))r})\] where $\tau'=\hfnc_1(\hat{\ct},\ct_0,\ct_1)$. 
    \end{itemize}
    }
    We argue that such a change is ``invisible'' to any $\ppt$ adversary due to the subgroup decision assumption $(\dgho)$ i.e., $\abslt{\prob{\win{7,i,1}}-\prob{\win{7,i-1,3}}}\leq \Adv{\AB}{\dgho}$.
    We describe a reduction informally to provide a proof sketch here.
    Given a $\dgho$ problem instance, 
    $\AB$ chooses $\amsk,{\tbmsk},a,b,\delta\sample \ZN$ to compute $\mpk$ similar to $\game{6}$.
    The $i^{th}$ decryption query uses the target $T$ of $\dgho$ problem instance to simulate the consistency check. In particular, the consistency check for $i^{th}$ decryption query is simulated by $e(\ct_2, T)\cdot e(\ct_1, T^{\bmsk})\neq e(\ct_0, T^{(a+\tau'b+P_{\YSet}(\amsk))})$ where $\tau'=\hfnc_1(\hat{\ct},\ct_0,\ct_1)$
    If $T$ was from $\sbGo$, the {\color{black}consistency check} is normal whereas if $T$ was from $\sbGot$, the {\color{black}consistency check} is semi-functional type-1.

    \item[$\game{7,i,2}$.] Here, the decryption consistency check of $i^{th}$ decryption query is further modified. 
    In particular, the consistency check is defined to be: Abort if 
    \[
    e(\ct_2^{(i)}, \sbgo^r\sbgt^r)\cdot e(\ct_1^{(i)},\sbgo^{\bmsk r}\sbgt^{\bmsk r})\neq e(\ct_0^{(i)}, \sbgo^{(a+\tau_i b+P_{\YSet}(\amsk))r}\sbgt^{\delta})\] 
    where $\tau_i=\hfnc_1(\hat{\ct}^{(i)},\ct_0^{(i)},\ct_1^{(i)})$. {We say that the {\color{black}consistency check} is semi-functional type-2.}
    Next, we argue that this change is invisible to any $\ppt$ adversary.
    Let us look at the semi-functional exponent of challenge ciphertext $\Ct=(\hat{\ct},\ct_0, \ct_1,\ct_2)$ and the ciphertext $\Ct_i=(\hat{\ct}^{(i)},\ct_0^{(i)}, \ct_1^{(i)},\ct_2^{(i)})$ in $i^{th}$ decryption query. Here, we mention that due to the Chinese Remainder Theorem, semi-functional exponents are independent of their normal counterparts. In particular, suppose $\ct_0=\sbgo^{s}\sbgt^{s}$ then, $s\mod{p_2}$ is independent to $s\mod{p_1}$.
    To prove our claim, we start with the natural restriction $\Ct\neq \Ct_i$. 
    Thus, 
    \begin{itemize}
        \item $(\hat{\ct},\ct_0,\ct_1)= (\hat{\ct}^{(i)},\ct_0^{(i)}, \ct_1^{(i)})$ and $\ct_2\neq \ct_2^{(i)}$: Due to equality of these three components, the hash value would be same. Thus $\tau=\hfnc_1(\hat{\ct},\ct_0,$ $\ct_1)=\hfnc_1(\hat{\ct}^{(i)},\ct_0^{(i)}, \ct_1^{(i)})=\tau_i$. This ensures $\ct_2= \ct_2^{(i)}$ which is a contradiction.
        \item $(\hat{\ct},\ct_0,\ct_1)\neq (\hat{\ct}^{(i)},\ct_0^{(i)}, \ct_1^{(i)})$: In this case, $\tau=\hfnc_1(\hat{\ct},\ct_0,\ct_1)\neq \hfnc_1(\hat{\ct}^{(i)},\ct_0^{(i)}, \ct_1^{(i)})=\tau_i$ due to collision resistance. Then, $(a+\tau b)\mod{p_2}$ is independent of $(a+\tau_i b)\mod{p_2}$ where $a,b\mod{p_2}$ are uniformly random choices from $\Zp[2]$. 
    \end{itemize}
    Therefore, changing the consistency check 
    \[
    e(\ct_2^{(i)}, \sbgo^r\sbgt^r)\cdot e(\ct_1^{(i)},\sbgo^{\bmsk r}\sbgt^{\bmsk r})\neq e(\ct_0^{(i)}, \sbgo^{(a+\tau_i b+P_{\YSet}(\amsk))r}\sbgt^{\delta})\] 
    with a random quantity $\delta$ will still stay invisible to the adversary.
    
    \item[$\game{7,i,3}$.] Here, the decryption consistency check of $i^{th}$ decryption query is modified.
    On $j^{th}$ decryption query on $(\XSet,\Ct,\YSet)$, 
    \begin{itemize}
        \item $j<i$: It samples $r\sample\ZN$ to define the consistency check to be: Abort if 
        \[\hspace{-0.3cm}e(\ct_2, \sbgo^r)\cdot e(\ct_1, \sbgo^{\bmsk r})\neq e(\ct_0, \sbgo^{(a+\tau'b+P_{\YSet}(\amsk))r}\sbgt^{\delta})\] where $\tau'=\hfnc_1(\hat{\ct},\ct_0,\ct_1)$ and $\delta\sample\ZN$ chosen during $\Setup$.
        \item $j=i$: It samples $\sbgo\sample\sbGo$ to define the consistency check to be:
        Abort if 
        \[\hspace{-0.3cm}e(\ct_2, \sbgo)\cdot e(\ct_1, \sbgo^{\bmsk}) \neq e(\ct_0, \sbgo^{(a+\tau'b+P_{\YSet}(\amsk))}\sbgt^{\delta})\] where $\tau_i=\hfnc_1(\hat{\ct}^{(i)},\ct_0^{(i)},\ct_1^{(i)})$ and $\delta\sample\ZN$ chosen during $\Setup$.
        \item $j>i$: It samples $r\sample\ZN$ to define the consistency check to be: 
        Abort if 
        \[\hspace{-0.3cm}e(\ct_2, \sbgo^r)\cdot e(\ct_1, \sbgo^{\bmsk r})\neq e(\ct_0, \sbgo^{(a+\tau'b+P_{\YSet}(\amsk))r})\] where $\tau'=\hfnc_1(\hat{\ct},\ct_0,\ct_1)$. 
    \end{itemize}
    We argue that such a change is ``invisible'' to any $\ppt$ adversary due to the subgroup decision assumption $(\dgho)$ i.e., $\abslt{\prob{\win{7,i,2}}-\prob{\win{7,i,3}}}\leq \Adv{\AB}{\dgho}$.
    We describe a reduction informally to provide a proof sketch here.
    Given a $\dgho$ problem instance, 
    $\AB$ chooses $\amsk,{\tbmsk},a,$ $b,\delta\sample \ZN$ to compute $\mpk$ similar to $\game{6}$.
    The $i^{th}$ decryption query uses the target $T$ of $\dgho$ problem instance to simulate the consistency check. In particular, the consistency check for $i^{th}$ decryption query is simulated by $e(\ct_2^{(i)}, T)\cdot e(\ct_1^{(i)}, T^{\bmsk})\neq e(\ct_0^{(i)}, T^{(a+\tau_i b+P_{\YSet}(\amsk))})$ where $\tau_i=\hfnc_1(\hat{\ct}^{(i)},\ct_0^{(i)},\ct_1^{(i)})$
    If $T$ was from $\sbGot$, the {\color{black}consistency check} is semi-functional type-2 whereas if $T$ was from $\sbGo$, the {\color{black}consistency check} is semi-functional type-3.
    \end{itemize}

    \item[$\game{8}$.] Here, we replace $\ssk=\hfnc(e(\ct_0, U_0))$ by a uniform random choice from $\Ssk$.
    The reason behind this is $U_0$ now is $u\cdot \sbgt^{z_0}\cdot R_{3}$.
    Due to $\game{6}$, $z_0$ is an uniformly random element independent of all $\{z_i\}_{i\in[q+m]}$.
    Thus $e(\ct_0, U_0)=e(\ct_0, u)\cdot e(\ct_0, \sbgt^{z_0})$ has $\log{p_2}$ bits of min-entropy due to $z_0\!\! \mod{p_2}$.
    Due to left-over hash lemma, $\ssk=\hfnc(e(\ct_0, U_0))$ is at most $2^{-{\secpp}}$ distant from uniform random choice from $\Ssk$ provided $\sbGt$ component in $\ct_0$ is not identity.
    The probability that the $\sbGt$ component of $\ct_0$ is $1$ is $1/{p_2}$. 
    Therefore $\abslt{\prob{\win{8}}-\prob{\win{7}}}\leq 1/p_2 + 2^{-{\secpp}}$.
    $\ssk$ now is a uniformly random element and it hides $\msg_{\bee}$ completely i.e. $\prob{\win{8}}=1/2$.
    \end{description}    
\end{proof}

\section{Applications} \label{sec:Applications}

\def\tosz{n}
\def\ttsz{n}
\def\pr{[\mathsf{P}]}

We now present \emph{the first} CCA-secure WIBE, CCA-secure WKD-IBE and CCA-secure DNF-ABE with constant-size secret keys and ciphertexts.
For this, we make use of the black-box transformations introduced by Katz \etal \cite{CANS:KMMS17}.
Chatterjee and Mukherjee \cite{RSA:ChaMuk19} applied these transformations on their constructions to achieve these well-known predicate encryptions as well.
However, as we already have mentioned, none of the existing constructions neither achieved constant-size ciphertext and secret key nor they achieved IND-CCA security.
We recall these transformations next and report the efficiency of the schemes.
Applying these transformations on $\spe$ (from \Cref{sec:Construct}) results in first CCA-secure WIBE, WKD-IBE and DNF-ABE with a constant-size secret key and ciphertext.

\paragraph{WIBE.}
Let the identity space is $\ID=\{0,1,*\}^n$. WIBE first converts an $n$-bit identity $\id$ into an $2n$-bit string $T_{\id}$.  Then we define $S_{\id}=\{i\in {[2n]}: T_{\id}[i]=1\}$.
Then $\Kgen(\id)$ and $\Enc(\id')$ of WIBE is simply $\spe.\Kgen$ and $\spe.\Enc$ operated on $S_{\id}$ and $S_{\id'}$ respectively.
\[
    T_{\id}[2i,2i+1]=
    \begin{cases}
        10 & \mbox{if } \id_i=1\\
        01 & \mbox{if } \id_i=0\\
        11 & \mbox{if } \id_i=*\\
    \end{cases}
\]

\paragraph{Wicked IBE.}
Let the identity space is $\ID=\{0,1,*\}^n$. WKD-IBE first converts an $n$-bit identity $\id$ into an $2n$-bit string $T_{\id}$.  Then we define $S_{\id}=\{i\in {[2n]}: T_{\id}[i]=1\}$.
Then $\Kgen(\id)$ and $\Enc(\id')$ of WKD-IBE is simply $\spe.\Kgen$ and $\spe.\Enc$ operated on $S_{\id}$ and $S_{\id'}$ respectively.
\[
    T_{\id}[2i,2i+1]=
    \begin{cases}
        10 & \mbox{if } \id_i=1\\
        01 & \mbox{if } \id_i=0\\
        00 & \mbox{if } \id_i=*\\
    \end{cases}
\]

\paragraph{DNF-ABE.}
Let $F=(C_1\vee C_2\vee \cdots C_{\gamma})$ be a DNF formula where each $C_i$ is conjunction of attributes i.e. $C_i=(I_{i,i_1}\wedge \ldots\wedge I_{i,i_k})$ where $I_{i,j}\in \ID$. 
For every $C_i$, we define $R_i=\{I_{i,i_1}, \ldots, I_{i,i_k}\}\subset \ID$.
Note that, an attribute set $\mathbb{A}$ satisfies the formula $F$ if $\exists k\in[{\gamma}] \text{ s.t. } C_k\subseteq \mathbb{A}\iff \text{ if } \exists k\in[{\gamma}] \text{ s.t. }  \ID\setminus \mathbb{A}\subseteq \ID\setminus C_k$.
Given an attribute set $\mathbb{A}$ and a DNF formula $F=(C_1\vee C_2\vee \cdots C_{\gamma})$, we output $S_Z=\{i\in \ID: T_Z[i]=1\}$ where $T_Z$ is defined as follows:
\[
T_Z[i]=
\begin{cases}
0 & \text{ if } i\in Z,\\
1  & \text{ if } i\notin Z\\
\end{cases}
\]
for all $Z\in\{C_1,\cdots, C_{\gamma},\mathbb{A}\}$.
The $\Kgen$ and $\Enc$ of DNF-ABE is simply $\spe.\Kgen$ and $\spe.\Enc$ running on such set $S_Z$ respectively.

\paragraph{Comparison of Parameters}
We compare our WIBE, WKD-IBE and DNF-ABE with the state of the art next \Cref{tab:WIBE-comparison,tab:WKDIBE-comparison,tab:ABE-comparison}.
In the comparisons next, $\tosz$ denotes the depth of hierarchy, $\ell$ is the bit-length of identity in Wat-IBE~\cite{EC:Waters05}, $\gamma$ is the number of disjunctive clauses in a DNF formula, and $\pr$ denotes the number of pairing operations.

\begin{table*}[htbp]
    \begin{center}
        \begin{scriptsize}
            \begin{tabular}{|c|c|c|c|c|c|c|}
                \hline
                WIBE Schemes & ${|\mpk|}$ & ${|\sk|}$ & ${|\Ct|}$ & $\Dec$ & Security & Assumption \\
                \hline
                BBG-WIBE~\cite{ICALP:ACDMNS06} & $(\tosz+4)\sbG$ & $(\tosz+2)\sbG$ & $(\tosz+2)\sbG+\sbGT$ & $2\pr$ & selective+CPA & $\tosz$-BDHI \\\hline
                Wat-WIBE~\cite{ICALP:ACDMNS06} & $((\tosz+1)\tosz+3)\sbG$ & $(\tosz+1)\sbG$ & $((\tosz+1)\tosz+2)\sbG+\sbGT$ & $(\tosz+1)\pr$ & adaptive+CPA & DBDH \\\hline
                SPE-1-WIBE~\cite{CANS:KMMS17} & $(2n+2)\Go+\GT$ & $\Gt+\Zp$ & $(2n+1)\Go+\GT$ & $1\pr$ & selective+CPA & $q$-BDHI \\\hline
                SPE-2-WIBE~\cite{CANS:KMMS17} & $(2n+1)\Go+2\Gt$ & $\Go+\Gt$ & $2n\Go+\Gt+\GT$ & $2\pr$ & selective+CPA & DBDH \\\hline
                SPE-WIBE \cite{RSA:ChaMuk19} & $(4\tosz+8)\Go+\GT$ & $5\Gt$ & $(2\tosz+2)\Go+\GT+2\tosz\ZZ_p$ & $3\pr$ & adaptive+CPA & SXDH \\\hline\hline
                Our {$\spe$}-WIBE & $(4\tosz+6)\sbG+\sbGT$ & $\sbG$ & $3\sbG+\sbGT$ & $5\pr$ & selective+CCA & SDP \\\hline
            \end{tabular}
        \end{scriptsize}
    \end{center}
    \caption{Comparison of efficient standard model WIBE schemes for $\tosz$-bit identity.} 
    \label{tab:WIBE-comparison}
\end{table*}
\begin{table*}[htbp]
    \begin{center}
        \begin{scriptsize}
            \begin{tabular}{|c|c|c|c|c|c|c|}
                \hline
                WKD-IBE Schemes & $|{\mpk}|$ & $|{\sk}|$ & $|{\Ct}|$ & $\Dec$ & Security & Assumption \\
                \hline
                BBG-WKD~\cite{ESORICS:AbdKilNev07} & $(\tosz+4)\sbG$ & $(\tosz+2)\sbG$ & $2\sbG+\sbGT$ & $2\pr$ & selective+CPA & $\tosz$-BDHE \\\hline
                Wat-WKD~\cite{ESORICS:AbdKilNev07} & $((\tosz+1)\tosz+3)\sbG$ & $(\tosz(\tosz+1)+2)\sbG$ & $2\sbG+\sbGT$ & $2\pr$ & adaptive+CPA & DBDH \\\hline
                SPE-1-WKD~\cite{CANS:KMMS17} & $(2n+2)\Go+\GT$ & $\Gt+\Zp$ & $(n+1)\Go+\GT$ & $1\pr$ & selective+CPA & $q$-BDHI \\\hline
                SPE-2-WKD~\cite{CANS:KMMS17} & $(2n+1)\Go+2\Gt$ & $\Go+\Gt$ & $n\Go+\Gt+\GT$ & $2\pr$ & selective+CPA & DBDH \\\hline
                SPE-WKD \cite{RSA:ChaMuk19} & $(4\tosz+8)\Go+\GT$ & $5\Gt$ & $(\tosz+2)\Go+\GT+\tosz\Zp$ & $3\pr$ & adaptive+CPA & SXDH \\\hline\hline
                Our {$\spe$}-WKD-IBE & $(4\tosz+6)\sbG+\sbGT$ & $\sbG$ & $3\sbG+\sbGT$ & $5\pr$ & selective+CCA & SDP \\\hline            
            \end{tabular}
        \end{scriptsize}
    \end{center}
    \caption{Comparison of efficient standard model WKD-IBE schemes for $\tosz$-bit identity.} 
    \label{tab:WKDIBE-comparison}
\end{table*}
\begin{table*}[htbp]
    \begin{center}
        \begin{scriptsize}
            \begin{tabular}{|c|c|c|c|c|c|c|}
                \hline
                DNF Schemes & $|{\mpk}|$ & $|{\sk}|$ & $|{\Ct}|$ & $\Dec$ & Security & Assumption  \\
                \hline
                SPE-1-DNF~\cite{CANS:KMMS17} & $(\ttsz+2)\Go+\GT$ & $\Gt+\Zp$ & $\gamma((\ttsz+1)\Go+\GT)$ & $1\pr$ & selective+CPA & $q$-BDHI \\\hline
                SPE-2-DNF~\cite{CANS:KMMS17} & $(\ttsz+1)\Go+2\Gt$ & $\Go+\Gt$ & $\gamma(2\ttsz\Go+\Gt+\GT)$ & $2\pr$ & selective+CPA & DBDH \\\hline
                SPE-DNF \cite{RSA:ChaMuk19} & $(2\ttsz+4)\Go+\GT$ & $5\Gt$ & $\gamma((\ttsz+2)\Go+\GT+\ttsz\Zp)$ & $3\pr$ & adaptive+CPA & SXDH \\\hline\hline
                Our {$\spe$}-DNF & $(2\ttsz+4)\sbG+\sbGT$ & $\sbG$ & $\gamma(3\sbG+\sbGT)$ & $5\pr$ & selective+CCA & SDP \\\hline
        \end{tabular}
        \end{scriptsize}
    \end{center}
    \caption{Comparison of efficient standard model DNF schemes. Here size of the universe is $\ttsz$ and $\gamma$ is the number of disjunctive clauses in a DNF formula.} 
    \label{tab:ABE-comparison}
\end{table*}

\section{Conclusion}\label{sec:Conclude}

This work proposes a large universe subset predicate encryption construction.
Our construction achieves constant-size ciphertext and secret key. 
We prove it to be selectively CCA-secure under the standard subgroup decision assumptions.
This work, thereby, closes an open problem mentioned in \cite{RSA:ChaMuk19}.
This work, also creates the first IBE with wildcards with constant ciphertext and keys thereby concluding this long-standing question on a positive note.
It is an interesting open problem to design a subset predicate encryption with similar efficiency as ours but achieves adaptively $\IndCCA$ security.


\appendix
\section{A Deferred Proof}\label{sec:Deferred}

\begin{proof}[Proof of \Cref{lem:SPE-I-k1-k0}]
    Let the reduction $\AB$ is given a problem instance $D=(\sbgo, \sbgr,$ $\sbgot, \sbgtr)$ and the target $T$.

    \begin{itemize}
    \item[Init.] The adversary $\AA$ sends the maximum set size $m$ and the challenger target set $\YSet$ s.t. $|\YSet|=\ell\leq m$.
    \item[Setup.] $\AB$ chooses $a,b,\amsk, \tbmsk\sample \ZN^2$ to compute the public parameters $\sbgo^\bmsk, \{G_i\}_{i\in[m]}$ where $\bmsk=\tbmsk\cdot P_{\YSet}(\amsk)\!\!\mod{N}$, $G_i = \sbgo^{\amsk^i}$. 
    $\AB$ chooses $\{\hkyr[j],\amsk_j\}_{j\in[k-1]}$ $\sample \ZN$ and $R'_{3,i}\sample \sbGr$ for $i\in[m]$ to define $U_i=T^{\amsk^i}\sbgtr^{\sum_{j\in[k-1]} \hkyr[j]\amsk_j^i} R'_{3,i}$.
    $\AB$ outputs the public parameter 
$\mpk=(\sbgo,\sbgo^\bmsk,\sbgo^a,\sbgo^b,\{G_i, U_i\}_{i\in[m]},$ 

\hfill$\sbgo^{\bmsk a},\sbgo^{\bmsk b},e(\sbgo,U_0)^\bmsk, \hfnc_1,\hfnc_2).$

    \item[Phase-I.] On a secret key query on $\XSet_i$, $\AB$ chooses $X'_3\sample \sbGr$ and sets $\sk_i = T^{\frac{\tbmsk\cdot P_{\YSet}(\amsk)}{P_{\XSet_i}(\amsk)}}\cdot \sbgtr^{\sum_{j\in[k-1]} \frac{\hkyr[j] \cdot\tbmsk\cdot P_{\YSet}(\amsk_j)}{P_{\XSet_i}(\amsk_j)}} X'_3.$
    
    \item[Challenge.] When $\AA$ provides the challenge $(\msg_0,\msg_1)$, $\AB$ samples $\bee\sample\{0,1\}$. $\AB$ returns $\Ct=(\hat{\ct},\ct_0, \ct_1,\ct_2)$ to $\AA$ where $\hat{\ct}=\msg_{\bee}\oplus\ssk$ where $\ssk=\hfnc_1(e(\ct_0, U_0))$, $\ct_0=\sbgot$, $\ct_1=\ct_0^{1/\tbmsk}$, $\ct_2=\ct_0^{(a+\tau b)}$ for $\tau=\hfnc_2(\hat{\ct},\ct_0,\ct_1)$ .
    \item[Phase-II.] Same as Phase-I queries.
    \item[Guess.] $\AB$ outputs $1$ if $\AA$'s guess $\bee'$ is same as $\AB$'s choice $\bee$.
    \end{itemize}

    \noindent If $T\in\sbGor$, then the game distribution is same as $\game{{5,k-1,1}}$.
    On the other hand, if $T\in\sbG$, then the game distribution is same as $\game{{5,k,0}}$ as can be seen in~\Cref{eqn:SPE-I-key-i}.
    \end{proof}

\bibliographystyle{./styl_files/IEEEtran}
\bibliography{./cryptobib/abbrev3,./cryptobib/crypto,./references}

\begin{thebibliography}{10}
\providecommand{\url}[1]{#1}
\csname url@samestyle\endcsname
\providecommand{\newblock}{\relax}
\providecommand{\bibinfo}[2]{#2}
\providecommand{\BIBentrySTDinterwordspacing}{\spaceskip=0pt\relax}
\providecommand{\BIBentryALTinterwordstretchfactor}{4}
\providecommand{\BIBentryALTinterwordspacing}{\spaceskip=\fontdimen2\font plus
\BIBentryALTinterwordstretchfactor\fontdimen3\font minus \fontdimen4\font\relax}
\providecommand{\BIBforeignlanguage}[2]{{%
\expandafter\ifx\csname l@#1\endcsname\relax
\typeout{** WARNING: IEEEtran.bst: No hyphenation pattern has been}%
\typeout{** loaded for the language `#1'. Using the pattern for}%
\typeout{** the default language instead.}%
\else
\language=\csname l@#1\endcsname
\fi
#2}}
\providecommand{\BIBdecl}{\relax}
\BIBdecl

\bibitem{C:BonGenWat05}
D.~Boneh, C.~Gentry, and B.~Waters, ``Collusion resistant broadcast encryption with short ciphertexts and private keys,'' in \emph{CRYPTO~2005}, ser. {LNCS}, V.~Shoup, Ed., vol. 3621.\hskip 1em plus 0.5em minus 0.4em\relax Springer, Heidelberg, Aug. 2005, pp. 258--275.

\bibitem{CANS:KMMS17}
J.~Katz, M.~Maffei, G.~Malavolta, and D.~Schr{\"o}der, ``Subset predicate encryption and its applications,'' in \emph{CANS 17}, ser. {LNCS}, S.~Capkun and S.~S.~M. Chow, Eds., vol. 11261.\hskip 1em plus 0.5em minus 0.4em\relax Springer, Heidelberg, Nov.~/~Dec. 2017, pp. 115--134.

\bibitem{ICALP:ACDMNS06}
M.~Abdalla, D.~Catalano, A.~Dent, J.~{Malone-Lee}, G.~Neven, and N.~Smart, ``Identity-based encryption gone wild,'' in \emph{ICALP 2006, Part~II}, ser. {LNCS}, M.~Bugliesi, B.~Preneel, V.~Sassone, and I.~Wegener, Eds., vol. 4052.\hskip 1em plus 0.5em minus 0.4em\relax Springer, Heidelberg, Jul. 2006, pp. 300--311.

\bibitem{ESORICS:AbdKilNev07}
M.~Abdalla, E.~Kiltz, and G.~Neven, ``Generalized key delegation for hierarchical identity-based encryption,'' in \emph{ESORICS~2007}, ser. {LNCS}, J.~Biskup and J.~L{\'o}pez, Eds., vol. 4734.\hskip 1em plus 0.5em minus 0.4em\relax Springer, Heidelberg, Sep. 2007, pp. 139--154.

\bibitem{RSA:ChaMuk19}
S.~Chatterjee and S.~Mukherjee, ``Large universe subset predicate encryption based on static assumption (without random oracle),'' in \emph{CT-RSA~2019}, ser. {LNCS}, M.~Matsui, Ed., vol. 11405.\hskip 1em plus 0.5em minus 0.4em\relax Springer, Heidelberg, Mar. 2019, pp. 62--82.

\bibitem{DSC:TseGao21}
Y.-F. Tseng and S.-J. Gao, ``Efficient subset predicate encryption for internet of things,'' in \emph{2021 IEEE Conference on Dependable and Secure Computing (DSC)}, 2021, pp. 1--2.

\bibitem{CSI:Tseng24}
\BIBentryALTinterwordspacing
Y.-F. Tseng, ``Attribute hiding subset predicate encryption: Quantum-resistant construction with efficient decryption,'' \emph{Computer Standards \& Interfaces}, vol.~88, p. 103796, 2024. [Online]. Available: \url{https://www.sciencedirect.com/science/article/pii/S0920548923000776}
\BIBentrySTDinterwordspacing

\bibitem{PKC:YAHK11}
S.~Yamada, N.~Attrapadung, G.~Hanaoka, and N.~Kunihiro, ``Generic constructions for chosen-ciphertext secure attribute based encryption,'' in \emph{PKC~2011}, ser. {LNCS}, D.~Catalano, N.~Fazio, R.~Gennaro, and A.~Nicolosi, Eds., vol. 6571.\hskip 1em plus 0.5em minus 0.4em\relax Springer, Heidelberg, Mar. 2011, pp. 71--89.

\bibitem{PKC:YASSHK12}
S.~Yamada, N.~Attrapadung, B.~Santoso, J.~C.~N. Schuldt, G.~Hanaoka, and N.~Kunihiro, ``Verifiable predicate encryption and applications to {CCA} security and anonymous predicate authentication,'' in \emph{PKC~2012}, ser. {LNCS}, M.~Fischlin, J.~Buchmann, and M.~Manulis, Eds., vol. 7293.\hskip 1em plus 0.5em minus 0.4em\relax Springer, Heidelberg, May 2012, pp. 243--261.

\bibitem{AAECC:NanPan17}
M.~Nandi and T.~Pandit, ``Verifiability-based conversion from {CPA} to {CCA}-secure predicate encryption,'' \emph{Applicable Algebra in Engineering, Communication and Computing}, vol.~29, no.~1, pp. 77--102, 2018.

\bibitem{EC:KatSahWat08}
J.~Katz, A.~Sahai, and B.~Waters, ``Predicate encryption supporting disjunctions, polynomial equations, and inner products,'' in \emph{EUROCRYPT~2008}, ser. {LNCS}, N.~P. Smart, Ed., vol. 4965.\hskip 1em plus 0.5em minus 0.4em\relax Springer, Heidelberg, Apr. 2008, pp. 146--162.

\bibitem{INDOCRYPT:ChaMukPan17}
S.~Chatterjee, S.~Mukherjee, and T.~Pandit, ``{CCA}-secure predicate encryption from pair encoding in prime order groups: Generic and efficient,'' in \emph{INDOCRYPT~2017}, ser. {LNCS}, A.~Patra and N.~P. Smart, Eds., vol. 10698.\hskip 1em plus 0.5em minus 0.4em\relax Springer, Heidelberg, Dec. 2017, pp. 85--106.

\bibitem{AMC:NanPan22}
\BIBentryALTinterwordspacing
M.~Nandi and T.~Pandit, ``Efficient fully cca-secure predicate encryptions from pair encodings,'' pp. 37--72, 2022. [Online]. Available: \url{/article/id/cd5ba434-d40f-4540-b194-8e0d89176196}
\BIBentrySTDinterwordspacing

\bibitem{EC:Attrapadung14}
N.~Attrapadung, ``Dual system encryption via doubly selective security: Framework, fully secure functional encryption for regular languages, and more,'' in \emph{EUROCRYPT~2014}, ser. {LNCS}, P.~Q. Nguyen and E.~Oswald, Eds., vol. 8441.\hskip 1em plus 0.5em minus 0.4em\relax Springer, Heidelberg, May 2014, pp. 557--577.

\bibitem{SCN:SVNHJ10}
S.~Sedghi, P.~{van Liesdonk}, S.~Nikova, P.~H. Hartel, and W.~Jonker, ``Searching keywords with wildcards on encrypted data,'' in \emph{SCN 10}, ser. {LNCS}, J.~A. Garay and R.~D. Prisco, Eds., vol. 6280.\hskip 1em plus 0.5em minus 0.4em\relax Springer, Heidelberg, Sep. 2010, pp. 138--153.

\bibitem{TCC:Wee16}
H.~Wee, ``D{\'e}j{\`a} {Q}: Encore! {U}n petit {IBE},'' in \emph{TCC~2016-A, Part~II}, ser. {LNCS}, E.~Kushilevitz and T.~Malkin, Eds., vol. 9563.\hskip 1em plus 0.5em minus 0.4em\relax Springer, Heidelberg, Jan. 2016, pp. 237--258.

\bibitem{SCN:GonLibRam18}
J.~Gong, B.~Libert, and S.~C. Ramanna, ``Compact {IBBE} and fuzzy {IBE} from simple assumptions,'' in \emph{SCN 18}, ser. {LNCS}, D.~Catalano and R.~{De Prisco}, Eds., vol. 11035.\hskip 1em plus 0.5em minus 0.4em\relax Springer, Heidelberg, Sep. 2018, pp. 563--582.

\bibitem{C:Waters09}
B.~Waters, ``Dual system encryption: Realizing fully secure {IBE} and {HIBE} under simple assumptions,'' in \emph{CRYPTO~2009}, ser. {LNCS}, S.~Halevi, Ed., vol. 5677.\hskip 1em plus 0.5em minus 0.4em\relax Springer, Heidelberg, Aug. 2009, pp. 619--636.

\bibitem{INDOCRYPT:BlaMuk20}
O.~Blazy and S.~Mukherjee, ``{CCA}-secure {ABE} using tag and pair encoding,'' in \emph{INDOCRYPT~2020}, ser. {LNCS}, K.~Bhargavan, E.~Oswald, and M.~Prabhakaran, Eds., vol. 12578.\hskip 1em plus 0.5em minus 0.4em\relax Springer, Heidelberg, Dec. 2020, pp. 691--714.

\bibitem{C:BonBoy04}
D.~Boneh and X.~Boyen, ``Secure identity based encryption without random oracles,'' in \emph{CRYPTO~2004}, ser. {LNCS}, M.~Franklin, Ed., vol. 3152.\hskip 1em plus 0.5em minus 0.4em\relax Springer, Heidelberg, Aug. 2004, pp. 443--459.

\bibitem{C:CraSho98}
R.~Cramer and V.~Shoup, ``A practical public key cryptosystem provably secure against adaptive chosen ciphertext attack,'' in \emph{CRYPTO'98}, ser. {LNCS}, H.~Krawczyk, Ed., vol. 1462.\hskip 1em plus 0.5em minus 0.4em\relax Springer, Heidelberg, Aug. 1998, pp. 13--25.

\bibitem{TCC:Wee14}
H.~Wee, ``Dual system encryption via predicate encodings,'' in \emph{TCC~2014}, ser. {LNCS}, Y.~Lindell, Ed., vol. 8349.\hskip 1em plus 0.5em minus 0.4em\relax Springer, Heidelberg, Feb. 2014, pp. 616--637.

\bibitem{EC:ChaMei14}
M.~Chase and S.~Meiklejohn, ``D{\'e}j{\`a} {Q}: Using dual systems to revisit q-type assumptions,'' in \emph{EUROCRYPT~2014}, ser. {LNCS}, P.~Q. Nguyen and E.~Oswald, Eds., vol. 8441.\hskip 1em plus 0.5em minus 0.4em\relax Springer, Heidelberg, May 2014, pp. 622--639.

\bibitem{EC:Waters05}
B.~R. Waters, ``Efficient identity-based encryption without random oracles,'' in \emph{EUROCRYPT~2005}, ser. {LNCS}, R.~Cramer, Ed., vol. 3494.\hskip 1em plus 0.5em minus 0.4em\relax Springer, Heidelberg, May 2005, pp. 114--127.

\end{thebibliography}

\end{document}